\documentclass[aap]{imsart}

\RequirePackage{amsthm,amsmath,amsfonts,amssymb}
\RequirePackage[numbers]{natbib}
\RequirePackage[colorlinks,citecolor=blue,urlcolor=blue]{hyperref}
\RequirePackage{graphicx}

\usepackage{amsmath}
\usepackage{amstext}
\usepackage{amssymb}
\usepackage{mathtools}
\usepackage{amsfonts}

\usepackage[english]{babel}
\usepackage{verbatim}
\usepackage{color}

\usepackage{tcolorbox}
\usepackage{tikz}

\usepackage{pgfplots}
\pgfplotsset{compat=1.16}
\usepgfplotslibrary{groupplots}
\usepackage{cancel}
\usepackage{float}

\usetikzlibrary{graphs,decorations.pathmorphing,decorations.markings} 
\usetikzlibrary{arrows}

\tikzset{degil/.style={
            decoration={markings,
            mark= at position 0.5 with {
                  \node[transform shape] (tempnode) {$\backslash$};
                  }
              },
              postaction={decorate}
}
}

\startlocaldefs
\theoremstyle{plain}
\newtheorem{theorem}{Theorem}[section]
\newtheorem{lemma}[theorem]{Lemma}
\theoremstyle{remark} 
\newtheorem{assumption}{Assumption}
\newtheorem{remark}{Remark}

\newtheorem{definition}{Definition}

\endlocaldefs

\def\rL{\mathcal{L}}
\def\rF{\mathbb{F}}
\def\R{\mathbb{R}}

\def\conv{\mathop{\rm conv}}

\def\Car{\mathop{\rm Car}}

\def\B{{\mathcal B}}

\def\H{{\mathcal H}}
\def\I{{\mathcal I}}
\def\G{{\mathcal G}}
\def\C{{\mathcal C}}

\def\P{{\mathcal P}}

\def\Y{{\mathcal Y}}
\def\M{{\mathcal M}}

\def\cE{\mathbb{E}}

\def\U{{\mathcal U}}

\newcommand{\cl}{\mathop{\rm cl}}
\newcommand{\cspan}{\mathop{\rm span}}

\def\abss{{\mathsf{abs}}}
\def\reff{{\mathsf{ref}}}

\def\ess{{\mathrm{ess}}}

\def\sUnif{{\mathsf{Unif}}}

\def\sX{{\mathsf X}}
\def\sY{{\mathsf Y}}
\def\sA{{\mathsf A}}
\def\sH{{\mathsf H}}
\def\sY{{\mathsf Y}}
\def\sB{{\mathsf B}}
\def\sM{{\mathsf M}}
\def\sE{{\mathsf E}}

\def\sR{{\mathsf R}}
\def\sU{{\mathsf U}}

\def\sZ{{\mathsf Z}}

\DeclareMathOperator*{\esssup}{ess\,sup}

\newcommand{\sy}[1]{{\color{black} #1}}

\allowdisplaybreaks

\begin{document}

\begin{frontmatter}
\title{Kernel Mean Embedding Topology: Weak and Strong Forms for Stochastic Kernels and Implications for Model Learning}
\runtitle{Kernel Mean Embedding Topologies for Stochastic Kernels}

\begin{aug}
\author[A]{\fnms{Naci}~\snm{Saldi}\ead[label=e1]{naci.saldi@bilkent.edu.tr}}
\author[B]{\fnms{and Serdar}~\snm{Y\"{u}ksel}\ead[label=e2]{yuksel@queensu.ca}}
\address[A]{Department of Mathematics,
Bilkent University\printead[presep={,\ }]{e1}}
\address[B]{Department of Mathematics and Statistics,
Queen's University\printead[presep={,\ }]{e2}}
\end{aug}

\begin{abstract}
We introduce a novel topology, called Kernel Mean Embedding Topology, for stochastic kernels, in a weak and strong form. This topology, defined on the spaces of Bochner integrable functions from a signal space to a space of probability measures endowed with a Hilbert space structure, allows for a versatile formulation. This construction allows one to obtain both a strong and weak formulation. (i) For its weak formulation, we highlight the utility on relaxed policy spaces, and investigate connections with the Young narrow topology and Borkar (or \( w^* \))-topology, and establish equivalence properties. We report that, while both the \( w^* \)-topology and kernel mean embedding topology are relatively compact, they are not closed. Conversely, while the Young narrow topology is closed, it lacks relative compactness. (ii) We show that the strong form provides an appropriate formulation for placing topologies on spaces of models characterized by stochastic kernels with explicit robustness and learning theoretic implications on optimal stochastic control under discounted or average cost criteria. (iii) We thus show that this topology possesses several properties making it ideal to study optimality and approximations (under the weak formulation) and robustness (under the strong formulation) for many applications.
\end{abstract}



\end{frontmatter}

\section{Weak and Strong Topologies on Stochastic Kernels}\label{sec1}

The set of stochastic kernel from one Borel space $\sY$ into another one $\sU$, is the collection of measurable mappings from $\sY$ into the set of probability measures $\P(\sU)$ on $\sU$:
\begin{align*}
\Gamma = \bigg\{\gamma: \text{$\gamma$ is a measurable function from $\sY$ to $\P(\sU)$} \bigg\}, 
\end{align*}
where $\P(\sU)$ is endowed with the Borel $\sigma$-algebra generated by the weak topology. 

In many areas of mathematics, engineering, and natural and applied sciences, such stochastic kernels play a crucial role. As essential components in probabilistic models, they are fundamental in areas such as Markov processes, controlled Markov processes, decision sciences, information theory, and learning theory. Understanding the topological properties of stochastic kernels is vital for analyzing their compactness, convergence, and approximation and learning-- key aspects that influence the study of stochastic dynamical systems and decision-making processes.

Our primary area of application will be control theory: In both deterministic and stochastic control theory, such stochastic kernels play a crucial role, representing the dynamics of a stochastic system (\cite[Lemma 1.2]{gihman2012controlled}, \cite[Lemma 3.1]{BorkarRealization}, or  \cite[Lemma F]{aumann1961mixed}) or the sets of relaxed/randomized control policies (\cite{young1937generalized},\cite{fleming1976generalized,kushner2001numerical,kushner2014partial},\cite{borkar1989topology}). 

Several topologies have been studied in the literature on $\Gamma$. Some topologies focus on pointwise, or almost everywhere (with respect to a given measure on $\sY$) convergence properties or uniform convergence properties. Some topologies, however, do not require pointwise convergence, which may be too demanding (especially in an empirical learning theoretic context). A review will be provided later in the paper.


\subsection{Summary of Main Results}

In this paper, we propose a mathematical framework for kernel mean embedding topologies on stochastic kernels viewed as Bochner integrable functions from a base space to a Hilbert space. We introduce such a topology in a strong and weak form. 

We establish several equivalence and relationship properties with classical and recently introduced topologies, as summarized in Figure~\ref{Hier}. Finally, we explore the implications of these topologies for model learning, empirical learning, and robustness in Markov decision theory context. Using Figure~\ref{Hier} as a summary, in the following we review the proposed topologies alongside the others discussed.

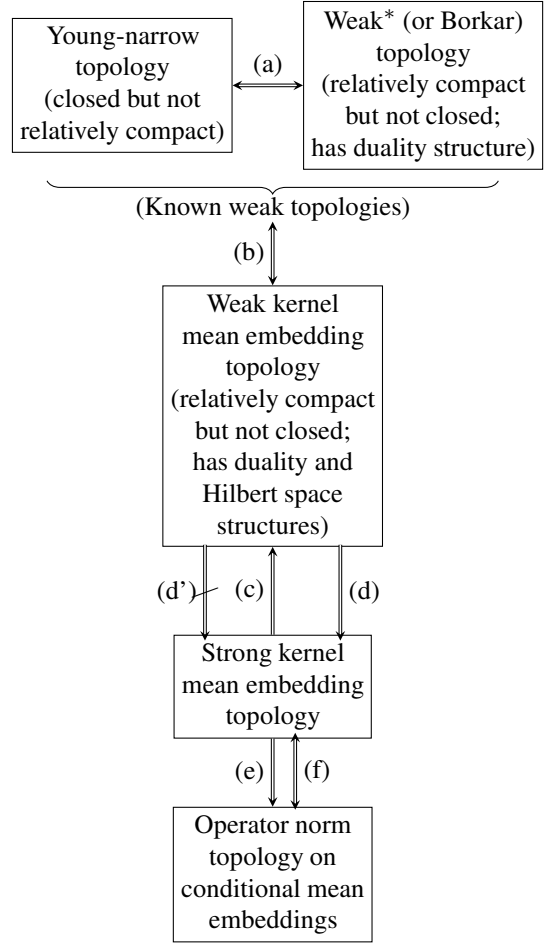
\begin{figure}[h]
\centering
\begin{tikzpicture}[
    >=stealth,
    node distance=1.5cm and 4cm,
    every node/.style={align=center, font=\small},
    ovalbox/.style={draw, shape=rectangle, minimum height=1cm, minimum width=2.5cm, text centered},
    arrow/.style={-latex}
]

\node[ovalbox] (young) at (-2, 8) {Young-narrow \\ topology \\ (closed but not \\ relatively compact)};
\node[ovalbox] (borkar) at (0, 3.6) {Weak kernel \\ mean embedding \\ topology \\ (relatively compact \\ but not closed; \\ has duality and \\ Hilbert space \\ structures)};
\node (borkar2) at (0.9, 2) { };
\node (borkar3) at (0, 6.3) { };
\node (borkar-1) at (-0.9, 2) { };
\node[ovalbox] (kernel) at (2, 8) {Weak$^*$ (or Borkar) \\ topology \\ (relatively compact \\ but not closed; \\ has duality structure)};

\node[ovalbox] (strong_kernel) at (0, 0) {Strong kernel \\ mean embedding \\ topology};
\node (strong_kernel2) at (0.9, 0.5) { };
\node (strong_kernel-1) at (-0.9, 0.5) { };

\node (strong_kernel3) at (0.3, -0.5) { };

\node[ovalbox] (operator) at (0, -2.5) {Operator norm \\ topology on \\ conditional mean \\ embeddings};
\node (operator2) at (0.3, -1.75) { };
\draw [decorate,decoration={brace,amplitude=5pt,mirror,raise=4ex}]
  (-3,7.3) -- (3,7.3) node[midway,yshift=-3em]{(Known weak topologies)};
\draw[<->, double] (young) -- node[midway, above] {(a)} (kernel);
\draw[<->, double] (borkar3) -- node[midway, left] {(b)} (borkar);
\draw[<-, double] (borkar) -- node[midway, left] {(c)} (strong_kernel);
\draw[->, double] (borkar2) -- node[midway, right] {(d)} (strong_kernel2);
\draw[->, double,degil] (borkar-1) -- node[midway, left] {(d')} (strong_kernel-1);
\draw[->, double] (strong_kernel) -- node[midway, left] {(e)} (operator);
\draw[<->, double] (strong_kernel3) -- node[midway, right] {(f)} (operator2);
\end{tikzpicture}
\caption{Hierarchy of topologies on stochastic kernels and relations with the topologies introduced in our paper: The strong kernel mean embedding topology and weak kernel mean embedding topology: (a) The Young-narrow and Borkar topologies are known in the literature, and their equivalence on the set of stochastic kernels is studied in \cite{yuksel2023borkar}. (b) In this paper, we introduce the weak kernel mean embedding topology and demonstrate its equivalence to the Young-narrow and Borkar topologies (see Theorem~\ref{equivalence}). (c) We also introduce a strong version of the kernel mean embedding topology, which has been used to establish empirical consistency results for approximating the conditional mean embeddings of stochastic kernels \cite{TaBa24,MaPe23}. It clearly implies the weak version of the kernel mean embedding topology. (d') As established in Theorem~\ref{strongStrongerthanWeakThm}, convergence in the weak kernel mean embedding topology does not imply convergence in the strong version. (d) However, as shown in Theorem~\ref{WeakImpliesStrong}, under the assumption of bounded and equicontinuous densities for the probability measures in the stochastic kernels, we can establish the implication from the weak to the strong version of the kernel topology. (e) {\color{black} Using conditional mean embedding framework \cite{SoHuSmFu09}, one can view stochastic kernels as operators between RKHS and $L_2$-space \cite{KlScSu20} with operator norm topology. \cite[Theorem 3.6]{MaPe23} shows the strong kernel mean embedding topology \sy{for the case with $q=2$} dominates this topology. (f) Alternatively, stochastic kernels can be viewed as Hilbert-Schmidt operators between an RKHS and an $L_2$-space, where the Hilbert-Schmidt norm topology, which is stronger than operator norm topology, is equivalent to the strong kernel mean embedding topology \sy{for the case with $q=2$} \cite[Theorem 5.1]{MaPe23}.}
}
\label{Hier}
\end{figure}

\subsubsection{A Weak Formulation}

We introduce our weak formulation in Definition \ref{defnWeakKMET}. Related to this formulation, two commonly studied topologies on $\Gamma$ are the following.

\noindent{\bf The Young Narrow Topology.} Since Young's seminal paper \cite{young1937generalized}, a particularly prominent approach has involved studying the topologies on Young measures defined by stochastic kernels. In this framework, stochastic kernels are identified with probability measures defined on a product space, maintaining a fixed marginal in the input space -- typically taken to be the Lebesgue measure in the context of optimal deterministic control (see \cite{young1937generalized,mcshane1967relaxed},\cite[Section 2.1]{castaing2004young},\cite[p. 254]{warga2014optimal},\cite{mascolo1989relaxation}, \cite[Theorem 2.2]{balder1988generalized}). 

Thus, under the Young narrow topology, one associates with $\Gamma$ the set of  probability measures induced on the product space $\sY \times \sU$ with a fixed marginal $\mu$ on $\sY$. On this product space, several weak topologies and their equivalence properties have been studied in the literature; see e.g, \cite{balder1997consequences}, \cite[Theorem 2.2]{balder1988generalized} and \cite[Theorem 3.32]{florescu2012young}.

The generalization to stochastic control problems to encompass more diverse input measures has become a common practice, especially in applications involving partially observed and decentralized stochastic control. Many of the references have adopted this approach to define control topologies. For example, this methodology has been used in studies on existence and approximation results related to optimal stochastic control \cite{fleming1976generalized,kushner2001numerical,kushner2014partial}, decentralized stochastic control \cite[Theorem 5.2]{YukselWitsenStandardArXiv} and \cite[Section 4]{SaYu22}, piecewise deterministic optimal stochastic control \cite{bauerle2010optimal,bauerle2018optimal}, economics and game theory \cite{milgrom1985distributional,mertens2015repeated,balder1988generalized}, mean-field control policies \cite{bayraktar2022finite}, and optimal quantization \cite[Definition 2.1]{YukselOptimizationofChannels}. \\

\noindent{\bf The $w^*$ or Borkar Topology.} Another notable topology on stochastic kernels is the one introduced by Borkar for randomized controls \cite{borkar1989topology}. This topology, which is further elaborated upon in \cite[Section 2.4]{arapostathis2012ergodic} and builds on the work in \cite{bismut1973theorie}, is formulated as a $w^*$-topology on stochastic kernels. In this framework, policies are treated as mappings from states or measurements to the space of signed measures with bounded variation, denoted as $\M(\sU)$, where probability measures $\P(\sU)$ form a subset. We also refer the reader to \cite{dieudonne1951theoreme,fattorini1994existence} for further references on such a $w^*$-formulation on stochastic kernels, in particular when instead of countably additive signed measures, finitely additive such measures are also considered.

References \cite{borkar1989topology,arapostathis2010uniform,pradhan2022near} offer a comprehensive analysis of various implications of $w^*$-topology in stochastic control theory, such as the continuity of expected costs in control policies \cite{borkar1989topology}, approximation results \cite{pradhan2022near} under different cost criteria, and the continuity of invariant measures of diffusions in control policies \cite{arapostathis2010uniform}.

Under the Borkar topology, one studies the space $\Gamma$ using a $w^*$-topology formulation as a bounded subset of the set of maps from $\sY$ to the space of signed measures with finite variation, which can be viewed as the topological dual of continuous functions that vanish at infinity. Therefore, according to the Banach-Alaoglu theorem \cite[Theorem 5.18]{Fol99}, this leads to a relatively sequentially compact metric space. Specifically, the unit ball of $L_{\infty}(\mu, \M(\sU)) = \rL_{1}(\mu, C_0(\sU))^*$ is compact under the $w^*$-topology, resulting in a relatively compact metric topology on stochastic kernels. It is important to note that the presentations in \cite[Section 3]{borkar1989topology} and \cite[Section 2.4]{arapostathis2012ergodic} differ slightly, although the induced topology remains identical. An equivalent representation of this topology is provided in \cite[Lemma 2.4.1]{arapostathis2012ergodic} (see also \cite[Lemma 3.1]{borkar1989topology}).

While $\sY$ was taken to be $\R^n$ in \cite{borkar1989topology}, Saldi extended this approach to scenarios where $\sY$ is a general standard Borel space equipped with a fixed input (probability) measure. This topology has also been employed by Saldi for optimal decentralized stochastic control in the context of existence analysis \cite{Sal20}. 

Recently, \cite{yuksel2023borkar} obtained equivalence results between the two topologies and showed that the average cost in optimal stochastic control problems changes continuously in policy under these topologies, leading to existence and approximation results. 

For our weak formulation, we will obtain several structural properties and investigate its connections with the Young narrow topology and Borkar (or \( w^* \))-topology and establish equivalence properties with both the Young narrow topology and the \( w^* \)-topology on the set of stochastic kernels. We demonstrate that the new topology we introduce is, in fact, equivalent to both the Young narrow topology and the $w^*$-topology on the set of stochastic kernels. This equivalence will be rigorously defined later in the paper, but the idea is that while these topologies may appear slightly different in their construction, they share fundamental properties when applied to stochastic kernels. Such a result is significant because each of these topologies offers distinct advantages in various contexts. For example, based on the results we present, both the $w^*$-topology and the kernel mean embedding topology are relatively (sequentially) compact, which is a crucial property when proving existence results. However, they are not closed. In contrast, the Young narrow topology is closed but lacks relative (sequential) compactness. Interestingly, to make the $w^*$-topology and the kernel mean embedding topology closed, a tightness condition is required -- 
this ensures uniform concentration of probability measures on compact sets. Similarly, Prokhorov's theorem shows that the same tightness condition is necessary for the Young narrow topology to achieve relative compactness. Thus, in terms of proving existence results via Weierstrass's extreme value theorem, these three topologies are equally advantageous. However, the $w^*$-topology offers a valuable duality structure that can be leveraged for certain analyses. In addition, the kernel mean embedding topology possesses a Hilbert space structure, which can facilitate approximations of stochastic kernels via simulations, as shown in the work \cite{MaPe23}.

Despite the differences in their individual properties, these three topologies -- $w^*$-topology, the weak kernel mean embedding topology, and Young narrow topology -- are topologically equivalent on the set of stochastic kernels. This does not contradict with the above discussion, as completeness is not a topological property.  A formal and rigorous statement of this result will follow after we introduce the necessary background (see Theorem~\ref{equivalence}).

\begin{theorem}\label{equivalence0}
Let $\{\gamma_{\lambda}\}$ be a sequence of stochastic kernels from one Borel space to another (locally compact) Borel space and let $\gamma$ also be a stochastic kernel. Then, the following are equivalent:
\begin{itemize}
\item[(i)] $\gamma_{\lambda} \rightharpoonup^* \gamma$ with respect to kernel mean embedding topology.
\item[(ii)] $\gamma_{\lambda} \rightharpoonup^* \gamma$ with respect to $w^*$-topology.
\item[(ii)] $\gamma_{\lambda} \rightharpoonup^* \gamma$ with respect to Young narrow topology.
\end{itemize}
\end{theorem}

In Theorem~\ref{equivalence0}, the key assumption is that $\gamma$ is also a stochastic kernel. If we drop this assumption and only require the sequence $\{\gamma_{\lambda}\}$ is convergent in the given topologies, the topologies are no longer equivalent. For example, in the $w^*$-topology or the kernel mean embedding topology, the sequence may converge to sub-stochastic kernel, which is not the case in Young narrow topology. As a result, the equivalence of these topologies breaks down when extended to non-stochastic kernels.

\subsubsection{A Strong Formulation}

We define the strong formulation in Definition \ref{StrongKMET}, as spaces of Bochner integrable functions from a signal space to a space of probability measures endowed with a Hilbert space structure. 

We will show in Theorem \ref{strongStrongerthanWeakThm} that the strong formulation leads to more stringent conditions compared with the weak one. However, in Theorem \ref{WeakImpliesStrong}, we present a partial converse implication and in Section \ref{secStrongKernelRobustness} we will show that the strong form is an ideal formulation for placing topologies on spaces of models with explicit robustness and learning implications on optimal stochastic control under discounted or average cost criteria. In particular, the kernel mean embedding topology has a Hilbert space structure, which is particularly useful for approximating stochastic kernels through simulation data. 

We will demonstrate that the weak formulation is more natural for placing topologies on control policies without compromising optimality classes, whereas the latter strong formulation is ideal for studying dynamical systems or models: Note that for applications concerning the latter, a (typically unknown) model is present a priori and one wishes to learn or identify the model in a sufficiently strong sense to ensure e.g. empirical consistency or robustness. Whereas, for the former, one would like to have a model in which admissible functions or policies are contained in the space of policies defined under the space of kernels and yet this space is suitable for analysis (such as on compactness, continuity and approximability).

{\bf Relation with Conditional Mean Embeddings.} Studying stochastic kernels within the framework of reproducing kernel Hilbert spaces (RKHS) is closely tied to the concept of conditional mean embedding (CME). CME is a kernel-based approach that embeds conditional distributions into an RKHS, providing a flexible and nonparametric way to model such distributions \cite{SoHuSmFu09, SoFuGr13, SoBoSiGoSm10, PaMu20}. The central idea is to represent a stochastic kernel, exactly or approximately, given a fixed input probability measure using covariance and cross-covariance operators, both of which are Hilbert-Schmidt operators. 

Initially, stochastic kernels were assumed to act as operators between the RKHSs of the input and output spaces. However, recent findings \cite{KlScSu20} show that the input RKHS must be replaced by an $L_2$-space because the conditional expectation of an output function belongs to this $L_2$-space, not the input RKHS. This shift leads to an operator norm topology for stochastic kernels. Notably, the strong kernel mean embedding topology dominates the operator norm topology \cite[Theorem 3.6]{MaPe23}. Alternatively, stochastic kernels can be viewed as Hilbert-Schmidt operators between an RKHS and an $L_2$-space. In this context, the Hilbert-Schmidt topology is equivalent to the strong kernel mean embedding topology \cite[Theorem 5.1]{MaPe23}. While exact representations of stochastic kernels via covariance and cross-covariance operators are generally infeasible \cite{KlScSu20}, certain conditions enable such representations using pseudo-inverses and adjoints of the operators \cite[Theorem 5.8]{KlSpSu21}. When these conditions are not met, approximate representations can be constructed using Tikhonov-Phillips regularization techniques \cite[Section 7]{MaPe23}. 

These approximations rely on the strong kernel mean embedding topology for convergence guarantees, facilitate data-driven modeling of stochastic kernels and enable robust model-based learning.

Data-driven approximation problems have been examined in works such as \cite{MoMuSu24, ZhMeMoGr22}, which provide convergence guarantees under specific spectral assumptions on certain operators. However, these spectral conditions are typically challenging to verify using probabilistic concepts. In practice, more interpretable conditions, such as continuity (or Lipschitz continuity) of the stochastic kernel with respect to familiar topologies -- e.g., weak convergence, setwise convergence, Wasserstein norm, or total variation -- are more explicit. Additionally, most analyses in this area rely on the strong kernel mean embedding topology.


In this paper, we lay the groundwork for the kernel mean embedding topology on the set of stochastic kernels. Building on this foundation, our future goal is to develop data-driven approximations of stochastic kernels with convergence guarantees within the model-based learning framework. In Section~\ref{applications}, we present a preliminary analysis of how kernel mean embedding topologies can be applied to the robustness problem, which is closely connected to the data-driven approach. 

To ensure the paper is as self-contained as possible, in the next section, we review duality results for vector-valued functions, which are essential for constructing various topologies for stochastic kernels.

\subsection{Notation and Conventions}

For a metric space $\sE$, the Borel $\sigma$-algebra is denoted by $\mathcal{E}$. We let $C_b(\sE)$, $C_0(\sE)$, and $C_c(\sE)$ denote the set of all continuous and bounded functions on $\sE$, the set of all continuous real functions on $\sE$ vanishing at infinity, and the set of all continuous real functions on $\sE$ with compact support, respectively. We always assume that $\sE$ is locally compact when working with $C_0(\sE)$ and $C_c(\sE)$ to avoid any trivial cases. {\color{black}
We also note that when $\sE$ is compact, the spaces $C_b(\sE)$, $C_0(\sE)$, and $C_c(\sE)$ coincide. In this case, we simply write $C(\sE)$ to denote any of them.} Let $\M(\sE)$ and $\P(\sE)$ denote the set of all finite signed measures and probability measures on $\sE$, respectively. A sequence $\{\mu_n\}$ of finite signed measures on $\sE$ is said to converge with respect to total variation distance (see \cite{HeLa03}) to a finite signed measure $\mu$ if $ \lim_{n\rightarrow\infty} 2\sup_{D \in \mathcal{E}} |\mu_n(D) - \mu(D)|=0$. A sequence $\{\mu_n\}$ of finite signed measures on $\sE$ is said to converge weakly (see \cite{HeLa03}) to a finite signed measure $\mu$ if $\int_{\sE} g d\mu_n \rightarrow \int_{\sE} g d\mu$ for all $g \in C_b(\sE)$. A sequence $\{\mu_n\}$ of finite signed measures on $\sE$ is said to converge vaguely (see \cite{HeLa03}) to a finite signed measure $\mu$ if $\int_{\sE} g d\mu_n \rightarrow \int_{\sE} g d\mu$ for all $g \in C_0(\sE)$. Unless otherwise specified, the term `measurable' will refer to Borel measurability.

\section{Bochner Spaces and Their Duals}\label{sec0}

Let $(\sY,\Y,\mu)$ be a probability space, where $\sY$ is a Borel space. Let $(\sB,\|\cdot\|_{\sB})$ be a separable Banach space and let $\sB^*$ denote the topological dual of $\sB$ with the induced norm $\|\cdot\|_{\sB^*}$ which turns $\sB^*$ into a Banach space. The duality pairing between any $b \in \sB$ and any $b^* \in \sB^*$ is denoted by $\langle b^*,b \rangle$. Hence, for any $b^* \in \sB^*$, the mapping $\sB \ni b \mapsto \langle b^*,b \rangle \in \R$ is linear and continuous.

For any $1 \leq p < \infty$, $p$-Bochner integrable functions are defined as $L_p$-integrable maps from $(\sY,\Y,\mu)$ to $\sB$. 

Similar to the definition of measurable functions, we start with definition of simple functions. A function $f: \sY \rightarrow \sB$ is said to be simple if there exists $b_1,\ldots,b_n \in \sB$ and $E_1,\ldots,E_n \in \Y$ such that
$
f(y) = \sum_{i=1}^n b_i \, 1_{E_i}(y). \nonumber
$
Define the Bochner integral of $f$ with respect to $\mu$ as
$
\int_{\sY} f(y) \, \mu(dy) \triangleq \sum_{i=1}^n b_i \, \mu(E_i). \nonumber
$
A function $f:\sY \rightarrow \sB$ is  said to be strongly measurable, if there exists a sequence $\{f_n\}$ of simple functions with $\lim_{n\rightarrow\infty} \|f_n(x)-f(x)\|_{\sB} = 0$ $\mu$-almost-everywhere. The strongly measurable function $f$ is $p$-Bochner-integrable \cite{DiUh77} if $\int_{\sY} \|f(y)\|_{\sB}^p \, \mu(dy) < \infty$. Let $L_p\bigl(\mu,\sB\bigr)$ denote the set of all $p$-Bochner-integrable functions from $(\sY,\Y,\mu)$ to $\sB$ endowed with the norm
\begin{align}
\|f\|_p \triangleq \left( \int_{\sY} \|f(y)\|_{\sB}^p \, \mu(dy) \right)^{1/p}. \nonumber
\end{align}
With this norm, $L_p\bigl(\mu,\sB\bigr)$ is a separable Banach space as $\sB$ is assumed to be separable.  

We now identify the topological dual of $L_p\bigl(\mu,\sB\bigr)$ which is denoted by $L_p\bigl(\mu,\sB\bigr)^*$. We start with the definition of weakly$^*$-measurable functions from $\sY$ to $\sB^*$. A function $\gamma: \sY \rightarrow \sB^*$ is called weakly$^*$-measurable \cite[p. 18]{CeMe97} if the mapping $\sY \ni y \mapsto \langle \gamma(y), b \rangle \in \R$ is $\Y/\B(\R)$-measurable for all $b \in \sB$. Let $\rL\bigl(\mu,\sB^*\bigr)$ denote the set of all weakly$^*$-measurable functions. Then, we define the following subsets
\begin{align}
\rL_q\bigl(\mu,\sB^*\bigr) &\triangleq \biggl\{ \gamma \in \rL\bigl(\mu,\sB^*\bigr): \|\gamma\|_{q} \triangleq \left( \int_{\sY} \|\gamma(y)\|_{\sB^*}^q \, \mu(dy) \right)^{1/q} < \infty \biggr\}, \,\, 1 \leq q < \infty, \label{extra-eq1} \\
\rL_{\infty}\bigl(\mu,\sB^*\bigr) &\triangleq \biggl\{ \gamma \in \rL\bigl(\mu,\sB^*\bigr): \|\gamma\|_{\infty} \triangleq\ess \sup_{y \in \sY} \|\gamma(y)\|_{\sB^*} < \infty \biggr\}, \label{eq1}
\end{align}
where $\ess \sup$ is taken with respect to the measure $\mu$. Then, we have the following theorem.

\begin{theorem}\cite[Theorem 1.5.5]{CeMe97}\label{mainduality}
Let $q$ be the conjugate of $p$; i.e., $1/p + 1/q =1$. 
For any $\gamma \in \rL_q\bigl(\mu,\sB^*\bigr)$ and $f \in L_p\bigl(\mu,\sB\bigr)$, let
\begin{align}
T_{\gamma}(f) \triangleq \int_{\sY} \langle \gamma(y), f(y) \rangle \, \mu(dy). \nonumber
\end{align}
Then the map $\rL_q\bigl(\mu,\sB^*\bigr) \ni \gamma \mapsto T_{\gamma} \in L_p\bigl(\mu,\sB\bigr)^*$ is an isometric isomorphism from $\rL_q\bigl(\mu,\sB^*\bigr)$ to $L_p\bigl(\mu,\sB\bigr)^*$. Hence, we can identify $L_p\bigl(\mu,\sB\bigr)^*$ with $\rL_q\bigl(\mu,\sB^*\bigr)$. For any $f \in L_p\bigl(\mu,\sB\bigr)$ and $\gamma \in \rL_q\bigl(\mu,\sB^*\bigr)$, the duality pairing is given by
\begin{align}
\langle\langle \gamma , f \rangle\rangle \triangleq \int_{\sY} \langle \gamma(y), f(y) \rangle \, \mu(dy).\nonumber
\end{align}
\end{theorem}

If $\sB^*$ is reflexive, then  $L_p\bigl(\mu,\sB\bigr)^*$ can be identified with set of $q$-Bochner integrable functions $L_q\bigl(\mu,\sB^*\bigr)$ \cite[Theorem 4.2.26]{PaWi18}. Indeed, in this case, by \cite[Corollary 4.2.5]{PaWi18}, we have 
$$
L_q\bigl(\mu,\sB^*\bigr) = \rL_q\bigl(\mu,\sB^*\bigr);
$$
that is weakly$^*$ measurability is equivalent to strong measurability\footnote{\cite[Corollary 4.2.5]{PaWi18} states that weak measurability is equivalent to strong measurability if $B^*$ is separable. But if $\sB^*$ is reflexive, then $\sB$ is also reflexive. Hence $(\sB^*)^* = \sB$,  and so, weak$^*$ measurability is equivalent to weak measurability. Since $\sB$ is assumed to be separable, $\sB^*$ is also separable. Therefore, weak$^*$-measurability is equivalent to strong measurability.}.

\begin{remark}
Under a more standard weak topology on stochastic kernels, called the $w^*$-topology to be reviewed in the next section, $\sB^*$ is to be the set of finite signed measures over some locally compact Borel space with total variation norm, which is not reflexive. Hence, in this case, we need to consider weakly$^*$-measurable functions. In the second topology, called the weak kernel mean embedding topology, $\sB^*$ is to be the separable reproducing kernel Hilbert space, which is both reflexive and separable. Hence, in the latter case, we do not need to consider weakly$^*$-measurable functions. 
\end{remark}

By Theorem~\ref{mainduality}, we equip $\rL_q\bigl(\mu,\sB^*\bigr)$\footnote{We note that if $\sB^*$ is reflexive, then we replace $\rL_q\bigl(\mu,\sB^*\bigr)$ with $L_q\bigl(\mu,\sB^*\bigr)$ as they are equivalent in this case.} with weak$^*$-topology induced by $L_p\bigl(\mu,\sB\bigr)$; that is, it is the smallest topology on $\rL_q\bigl(\mu,\sB^*\bigr)$ for which the mapping
\begin{align}
\rL_q\bigl(\mu,\sB^*\bigr) \ni \gamma \mapsto \int_{\sY} \langle \gamma(y), f(y) \rangle \, \mu(dy) \in \R \nonumber
\end{align}
is continuous for all $f \in L_p\bigl(\mu,\sB\bigr)$. We write $\gamma_{\lambda} \rightharpoonup^* \gamma$, if $\gamma_{\lambda}$ converges to $\gamma$ in $\rL_q\bigl(\mu,\sB^*\bigr)$ with respect to weak$^*$-topology.

Suppose $G$ is a subset of $\sB^*$ and define
\begin{align}
\rL_q\bigl(\mu,G\bigr) \triangleq \biggl\{ \gamma \in \rL_q\bigl(\mu,\sB^*\bigr): \gamma(y) \in G \text{ } \mu-\text{almost-everywhere} \biggr\}. \nonumber
\end{align}
If $G$ is a closed and bounded subset, then $\rL_q\bigl(\mu,G\bigr)$ is also closed and bounded in $\rL_q\bigl(\mu,\sB^*\bigr)$, and so, by Banach-Alaoglu Theorem \cite[Theorem 5.18]{Fol99}, it is compact with respect to weak$^*$-topology. Since $L_p\bigl(\mu,\sB\bigr)$ is separable, by \cite[Lemma 1.3.2]{HeLa03}, $\rL_q\bigl(\mu,G\bigr)$ is metrizable, and so, is also sequentially compact.

\begin{remark}\label{equivalence-weakstar}
If \(G\) is closed and bounded, then \(\rL_q\bigl(\mu, G\bigr) = \rL_{q'}\bigl(\mu, G\bigr)\) for any \(q, q' \in [1, \infty]\). Hence, the choice of \(q\) is not significant. Moreover, the weak$^*$ topology on \(\rL_q\bigl(\mu, G\bigr)\) can be induced using any \(L_p\bigl(\mu, \sB\bigr)\) for \(1 \leq p < \infty\). In fact, for any \(1 \leq p < \infty\), there exists a common dense subset \(\mathcal{D}\) within each \(L_p\bigl(\mu, \sB\bigr)\). Therefore, if  
\[
\int_{\sY} \langle \gamma_{\lambda}(y), f(y) \rangle \, \mu(dy) \rightarrow \int_{\sY} \langle \gamma(y), f(y) \rangle \, \mu(dy)
\]  
for any \(f \in \mathcal{D}\), where \(\gamma_{\lambda}, \gamma \in \rL_q\bigl(\mu, G\bigr)\), then  
\[
\int_{\sY} \langle \gamma_{\lambda}(y), f(y) \rangle \, \mu(dy) \rightarrow \int_{\sY} \langle \gamma(y), f(y) \rangle \, \mu(dy)
\]  
for any \(f \in L_p\bigl(\mu, \sB\bigr)\), for all \(1 \leq p < \infty\). This conclusion follows from the boundedness of \(G\) (without boundedness it does not hold). Consequently, except for \(p = \infty\), where the other conjugate pair is \(q = 1\), the weak$^*$ topologies on \(\rL_q\bigl(\mu, G\bigr)\) induced by \(L_p\bigl(\mu, \sB\bigr)\) for all \(1 \leq p < \infty\) are equivalent. Hence, it does not matter which \((p, q)\) pair is chosen as long as \(1 \leq p < \infty\). Hence, for the rest of the paper, we introduce weak$^*$-topologies without specifying a particular conjugate pair \((p, q)\).
\end{remark}

\section{Weak Kernel Mean Embedding Topology on Stochastic Kernels}\label{kernel-mean}

In this section, we introduce a novel topology on the set of stochastic kernels, which we refer to as the weak kernel mean embedding topology. After we introduce the topology and discuss several properties, we will then review two relatively well studied related topologies and establish several equivalence properties.

\subsection{Weak Kernel Mean Embedding Topology on Stochastic Kernels}\label{kernel-mean}

In this section, we introduce the weak kernel mean embedding topology. This topology is intrinsically connected to the theory of reproducing kernel Hilbert spaces (RKHS). Given this connection, a basic understanding of RKHS is essential for grasping the underlying principles of the weak kernel mean embedding topology. For readers unfamiliar with the basics of RKHS, we recommend the comprehensive treatment provided in \cite{PaRa16}, which serves as an excellent introduction to the foundational concepts in RKHS; further recent review studies include \cite{ghojogh2021reproducing}.

{\color{black}
Let 
$$k:\sU \times \sU \rightarrow \R$$ be a positive semi-definite kernel that induces the reproducing kernel Hilbert space $\H_k$. Indeed, let us define the feature function $\Phi$ as 
$$
\Phi(u) \triangleq k(\cdot,u), \,\, \forall u \in \sU.
$$
Therefore, we have 
$$
\H_k = \cl \cspan\{\Phi(u): u \in \sU\},
$$
where the closure is taken with respect to the topology induced by the inner product
$$
\left \langle \sum_{i=1}^n \alpha_i \, \Phi(u_i), \sum_{j=1}^m \gamma_j \, \Phi(u_j) \right \rangle_{\H_k} \triangleq \sum_{i=1,j=1}^{n,m} \alpha_i \, \gamma_j \, k(u_i,u_j).
$$
This inner product has the reproducing property, that is, for any $f \in \H_k$, we have 
$$
f(u) = \langle f, \Phi(u) \rangle_{\H_k}.
$$
In particular, this implies 
$$
k(u_1,u_2) = \langle \Phi(u_1),\Phi(u_2) \rangle_{\H_k},
$$
which suggests that the kernel \( k \) can be interpreted as an inner product between two feature mappings in \( \mathcal{H}_k \).
}

We suppose that the following conditions are true for RKHS $\H_k$ and its kernel $k$. These assumptions are satisfied by well-known kernels like Gaussian, Laplacian, Matern, etc. (see \cite[Lemma 8]{SiBaScMa24}, \cite[Proposition 5.6]{CaDeToUm10}, \cite[Section 3.1]{BhKeGe11}). 

\begin{tcolorbox} 
[colback=white!100]
\begin{assumption}\label{as2}
\begin{itemize}
\item[ ]
\item[(a)] $k$ is bounded and continuous. 
\item[(b)] $k(\cdot,u) \in C_0(\sU)$ for all $u \in \sU$. 
\item[(c)] $\H_k$ is dense in $C_0(\sU)$. 
\end{itemize}
\end{assumption}
\end{tcolorbox}
Under the assumptions above, any finite signed measure $\nu \in \M(\sU)$ can be embedded into $\H_k$ as follows:
\begin{align}\label{defnIinject}
I: \M(\sU) \ni \nu \mapsto I_{\nu} \triangleq \int_{\sU} k(\cdot,u) \, \nu(du) \in \H_k,
\end{align}
where the integral on the right is in Bochner sense. In the literature, this concept is known as the \emph{kernel mean embedding of distributions}. For a detailed introduction to this topic, we refer the reader to the monograph \cite{MuFuSrSc17} and the paper \cite{SrBhGrFuScLa10}.

Kernel mean embedding is a powerful technique from machine learning and probability theory that provides a way to represent probability distributions in a high-dimensional feature space, which is associated with a reproducing kernel Hilbert space. Instead of working directly with the probability distribution, the method embeds the distribution as a single point in the RKHS by mapping it through a feature map $ \sU \ni u \mapsto k(\cdot,u) \in \H_k$. This embedding enables the manipulation and comparison of distributions using linear operations in the RKHS, making it especially useful in non-parametric statistics and machine learning tasks. For example, the kernel mean embedding method allows for tasks like estimating expectations, comparing distributions, and performing statistical inference by working with these embedded representations. A key advantage of kernel mean embedding is its flexibility to handle complex, high-dimensional data without making strong parametric assumptions about the underlying distribution. This approach is widely used in applications such as density estimation, non-parametric filtering, and reinforcement learning, where distributions or stochastic processes need to be approximated or estimated effectively.

Let $\H_{\M} \subset \H_k$ be the image of $\M(\sU)$ under $I$ and let $\H_{\P}$ be the image of $\P(\sU)$ under $I$. One can prove that the map $I$ is injective as $\H_k$ is dense in $C_0(\sU)$ \cite{BhKeGe11}, and so, we can uniquely identify any finite signed measure $\nu$ on $\sU$ via its image $I_{\nu}$. This embedding naturally induces a topology, called \emph{maximum mean discrepancy (MMD) topology}, on the set of signed measures as follows: 

\begin{tcolorbox} 
[colback=white!100]
\begin{definition}[Maximum Mean Discrepancy]
We have $\nu_{\lambda} \rightarrow \nu$ in $\M(\sU)$ with respect to the maximum mean discrepancy (MMD) topology if and only if $I_{\nu_{\lambda}} \rightarrow I_{\nu}$ in $\H_k$. 
\end{definition}
\end{tcolorbox}

Indeed, this topology on $\M(\sU)$ (or equivalently on $\H_{\M}$) is induced by the following inner product:
$$
\langle \nu_1,\nu_2 \rangle_{\M} \triangleq \langle I_{\nu_1},I_{\nu_2} \rangle_{\H_k} = \int_{\sU \times \sU} k(u_1,u_2) \, \nu_1(du_1) \, \nu_2(du_2). 
$$
\sy{where we use the fundamental RKHS property that $\langle k(\cdot,u_1), k(\cdot,u_2) \rangle_{\H_k}  = k(u_1,u_2)$.}
Hence, under this inner product, we can view $\M(\sU)$ as a pre-Hilbert space whose closure is $\H_k$. {\color{black} This inner product also induces the following norm, called MMD norm: 
\begin{align*}
\|\nu_1-\nu_2\|_{\M} &\triangleq \|I_{\nu_1} - I_{\nu_2} \|_{\H_k} \\
&= \sup_{\|h\|_{\H_k} \leq 1} \langle h,I_{\nu_1} - I_{\nu_2} \rangle_{\H_k} \\
&= \sup_{\|h\|_{\H_k} \leq 1} \left( \langle h,I_{\nu_1} \rangle_{\H_k} - \langle h,I_{\nu_2} \rangle_{\H_k} \right) \\
&= \sup_{\|h\|_{\H_k} \leq 1} \left( \left\langle h,\int_{\sU} k(\cdot,u) \, \nu_1(du) \right\rangle_{\H_k} - \left\langle h,\int_{\sU} k(\cdot,u) \, \nu_2(du) \right\rangle_{\H_k}\right) \\
&= \sup_{\|h\|_{\H_k} \leq 1} \left( \int_{\sU} \langle h,k(\cdot,u) \rangle_{\H_k} \, \nu_1(du)  - \int_{\sU} \langle h,k(\cdot,u) \rangle_{\H_k}  \, \nu_2(du) \right) \\
&= \sup_{\|h\|_{\H_k} \leq 1} \left(\int_{\sU} h(u) \, \nu_1(du) - \int_{\sU} h(u) \, \nu_2(du_2) \right). 
\end{align*}
}
This norm looks like a total variation norm:
$$
\|\nu_1-\nu_2\|_{TV} \triangleq \sup_{\|h\| \leq 1} \left(\int_{\sU} h(u) \, \nu_1(du) - \int_{\sU} h(u) \, \nu_2(du_2) \right)
$$
but their behavior is indeed quite different. To see this, we first note that if $h,g \in \H_k$, then 
\begin{align*}
\|h-g\| &\triangleq \sup_{u \in \sU} |h(u)-g(u)| \\
&= \sup_{u \in \sU} | \langle h-g,k(\cdot,u) \rangle_{\H_k} | \\
&\leq \sup_{\|l\|_{\H_k} \leq M} |\langle h-g,l \rangle_{\H_k}|, \,\, \text{where } \,\, M \triangleq \sup_{u \in \sU} \|k(\cdot,u)\|_{\H_k} = \sup_{u \in \sU} \sqrt{k(u,u)} < \infty \\
&= M \,  \sup_{\|l\|_{\H_k} \leq 1} |\langle h-g,l \rangle_{\H_k}| \\
&= M \, \|h-g\|_{\H_k}. 
\end{align*}
This implies that 
$$
\|\nu_1-\nu_2\|_{\M} \leq M \, \|\nu_1-\nu_2\|_{TV}. 
$$
Hence, the MMD norm is weaker than the total variation norm. In fact, on the set of probability measures $\P(\sU)$, the MMD norm is equivalent to the weak convergence topology, as demonstrated in \cite[Lemma 5]{SiBaScMa24}. This equivalence implies that the MMD topology is sufficiently weak to establish empirical consistency results as opposed to the total variation norm. Namely, if $\nu$ is a nonatomic measure, then for any singleton $\{u\}$, we have $\nu(\{u\}) = 0$. This means that for any $n$, if $\nu_n$ is the empirical estimate of $\nu$, then the total variation distance $\|\nu_n - \nu\|_{TV}$ will always be $2$. Consequently, total variation distance is not suitable for establishing empirical consistency. However, when considering weak convergence topology, the strong law of large numbers ensures that $\nu_n \rightarrow \nu$ almost surely as $n \rightarrow \infty$. Therefore, weak convergence topology is the appropriate notion of distance for empirical consistency. Since the MMD topology is equivalent to the weak convergence topology on the set of probability measures, MMD norm as opposed total variation norm can also be effectively used for this purpose. Moreover, the MMD topology endows the set of signed measures with an inner product structure, allowing for the use of geometrical arguments such as orthogonality and projection, which is missing in total variation norm.

In summary, above properties make the MMD topology particularly useful. Firstly, its equivalence to weak convergence topology ensures that it retains all the desirable attributes of weak convergence topology, including empirical consistency and the ability to handle convergence of sequences of measures under sufficiently weak conditions. This is particularly important in statistical applications where consistency and convergence are crucial.

Secondly, the inner product structure provided by the MMD topology opens up new avenues for analysis. With this structure, one can employ tools from functional analysis and geometry, such as orthogonality and projection, which are not typically available in other topologies on the space of measures. This can lead to more powerful and intuitive methods for dealing with problems in probability and statistics, particularly in high-dimensional spaces or in the context of machine learning algorithms where geometrical insights can provide significant advantages.

Note that by Riesz representation theorem, $\H_k^* = \H_k$. Hence, in this particular case, we have the following identifications \( \sB = \H_k \) and $\sB^* = \H_k$. Note that, in this case, $\sB^* = \H_k$ is reflexive. Therefore, we do not need to consider weakly$^*$-measurable functions when considering duality results. In view of this, we can define the following Banach spaces $L_p\bigl(\mu,\H_k\bigr)$ and $L_q\bigl(\mu,\H_k\bigr)$ with the following norms, respectively,
\begin{align}
\|f\|_p &= \left( \int_{\sY} \|f(y)\|^p_{\H_k} \, \mu(dy) \right)^{1/p}, \label{eq3} \\
\|\gamma\|_q &= \begin{cases}
\left( \int_{\sY} \|\gamma(y)\|^q_{\H_k} \, \mu(dy) \right)^{1/q}, & \,\, \text{if} \,\, q < \infty \\
\ess \sup_{y \in \sY} \|\gamma(y)\|_{\H_k}, & \,\, \text{if} \,\, q = \infty.
\end{cases}
\label{eq4}
\end{align}
Since $\H_k$ is separable, $L_p\bigl(\mu,\H_k\bigr)$ is also separable. By Theorem~\ref{mainduality}, the topological dual of $L_p\bigl(\mu,\H_k\bigr)$ is $L_q\bigl(\mu,\H_k\bigr)$; that is
\begin{align}
L_p\bigl(\mu,\H_k\bigr)^* = L_q\bigl(\mu,\H_k\bigr). \nonumber
\end{align}
For any $f \in L_p\bigl(\mu,\H_k\bigr)$ and $\gamma \in L_q\bigl(\mu,\H_k\bigr)$, the duality pairing is given by
\begin{align}
\langle\langle \gamma, f \rangle\rangle &= \int_{\sY} \langle \gamma(y), f(y) \rangle_{\H_k} \, \mu(dy). \nonumber 
\end{align}
Hence, we can equip $L_q\bigl(\mu,\H_k\bigr)$ with weak$^*$-topology induced by $L_p\bigl(\mu,\H_k\bigr)$. This topology is called \emph{weak kernel mean embedding topology}. Under this topology, $\gamma_{\lambda} \rightharpoonup^* \gamma$ in $L_q\bigl(\mu,\H_k\bigr)$, if
\begin{align}
&\int_{\sY} \langle \gamma_{\lambda}(y), f(y) \rangle_{\H_k} \, \mu(dy) \rightarrow \int_{\sY} \langle \gamma(y), f(y) \rangle_{\H_k} \, \mu(dy), \nonumber
\end{align}
for all $f \in L_p\bigl(\mu,\H_k\bigr)$.

Recall that $\Gamma$ denotes the set of all stochastic kernels from $\sY$ to $\sU$. Define
$$
I \circ \Gamma \triangleq \{I \circ \gamma: \gamma \in \Gamma\};
$$
that is, we embed $\gamma(y) \in \P(\sU)$ into $\H_k$ via $I$ for all $y$ \sy{(recall the definition of $I$ in (\ref{defnIinject}))}. 

\begin{lemma}\label{kernel-MMD}
The set $I \circ \Gamma$ is the same as the strongly measurable functions from $\sY$ to $\H_k$ whose range is contained in $\H_{\P}$. 
\end{lemma}

\begin{proof}
First, we note that since $\H_k$ is separable and reflexive, then the set of strongly measurable functions from $\sY$ to $\H_k$ is equivalent to the set of weakly$^*$ measurable functions from $\sY$ to $\H_k$. 

Let $\gamma \in \Gamma$. Then, by Lemma~\ref{kernel}, $\gamma$ is a weakly$^*$ measurable function from $\sY$ to $\M(\sU)$ whose range is contained in $\P(\sU)$. Hence, for any $g \in \H_k \subset C_0(\sU)$, the following mapping
$$
\sY \ni y \mapsto \langle \gamma(y),g \rangle = \int_{\sU} g(u) \, \gamma(y)(du) = \langle I \circ \gamma(y), g \rangle_{\H_k} \in \R
$$ 
is measurable. Hence, $I \circ \gamma$ is weakly$^*$ (and so strongly) measurable function from $\sY$ to $\H_k$ whose range is contained in $\H_{\P}$. 

Conversely, let $\xi$ be a weakly$^*$ (and so strongly) measurable function from $\sY$ to $\H_k$ whose range is contained in $\H_{\P}$. Then, define $\gamma \triangleq \xi \circ I^{-1}$ (as $I$ is injective, this inverse is well-defined). Then, for any $g \in \H_k \subset C_0(\sU)$, the following mapping
$$
\sY \ni y \mapsto \langle \xi(y),g \rangle_{\H_k} = \int_{\sU} g(u) \, \gamma(y)(du) = \langle \gamma(y), g \rangle \in \R
$$ 
is measurable. Since $\H_k$ is dense in $C_0(\sU)$, the following mapping
$$
\sY \ni y \mapsto \int_{\sU} g(u) \, \gamma(y)(du) = \langle \gamma(y), g \rangle \in \R
$$ 
is also measurable for any $g \in C_0(\sU) \setminus \H_k$ as well. Hence, $\gamma$ is a weakly$^*$ measurable function from $\sY$ to $\M(\sU)$ whose range is contained in $\P(\sU)$. {\color{black} As Lemma~\ref{kernel} will show, $\Gamma$ consists exactly of those weakly$^*$ measurable functions from $\sY$ to $\M(\sU)$ whose range lies in $\P(\sU)$.} 
\end{proof}

Hence, by Lemma~\ref{kernel-MMD}, the set of stochastic kernels can be identified with the following bounded subset $L_q(\mu,\H_{\P})$ of $L_q(\mu,\H_k)$ {\color{black} (Here, we let \( q \in (1, \infty] \), though in some parts we exclude the case \( q = \infty \) in order to make use of reflexivity. The case \( q = 2 \) is of particular interest because of the underlying Hilbert space structure.)}:
\begin{tcolorbox} 
[colback=white!100]
$$L_q(\mu,\H_{\P}) = I \circ \Gamma \equiv \Gamma.$$
\end{tcolorbox}
Let $\H_{\P_{_{\leq1}}}$ be the image of the set of sub-probability measures $\P{_{\leq1}}(\sU)$ under $I$. Since $\H_{\M}$ is a strict subset of $\H_k$, it is not immediately clear whether $\H_{\P_{_{\leq1}}}$ is a closed subset of $\H_k$. Similarly, it is not clear that $L_q\bigl(\mu, \H_{\P_{_{\leq1}}}\bigr)$ is closed in $L_q\bigl(\mu, \H_k\bigr)$ with respect to the weak kernel mean embedding topology. Therefore, we need to establish these properties first. 

\begin{lemma}\label{subprobability-closed}
$\H_{\P_{_{\leq1}}}$ is a closed subset of $\H_k$.
\end{lemma}

\begin{proof}
Let $I_{\nu_n} \rightarrow h$ in $\H_k$, where $\{\nu_n\} \subset \P_{_{\leq1}}$. Hence, for any $g \in \H_k$, we have 
$$
\lim_{n \rightarrow \infty} \langle I_{\nu_n}, g \rangle_{\H_k} = \lim_{n \rightarrow \infty} \int_{\sU} g(u) \, \nu_n(du) = \langle h, g \rangle_{\H_k}.
$$
Let us define the functional $L: \H_k \rightarrow \R$ as follows $L(g) \triangleq \langle h, g \rangle_{\H_k}$. This functional is linear. Moreover, since 
$$
\sup_n \sup_{g \in \H_k: \|g\| \leq 1} \left| \int_{\sU} g(u) \, \nu_n(du) \right| \leq 1,
$$
we have
\begin{align*}
\sup_{g \in \H_k: \|g\| \leq 1} |L(g)| \leq 1.
\end{align*}
Hence, $L$ is a bounded linear functional on $\H_k$, when $\H_k$ is endowed with sup-norm instead of $\H_k$-norm. Since $\H_k$ is dense in $C_0(\sU)$, by continuous linear extension theorem, we can uniquely extend $L$ to $C_0(\sU)$ as a bounded linear functional, denoted as $\hat L$. But since $C_0(\sU)^* = \M(\sU)$, there exists $\nu \in \M(\sU)$ such that $\hat L = \nu$; that is, for any $g \in C_0(\sU)$, we have $\hat L(g) = \int_{\sU} g(u) \, \nu(du)$. One can prove that $\hat L$ is a positive functional; that is, if $g \geq 0$, then $\hat L(g) \geq 0$. Hence, $\nu \in P_{_{\leq1}}$ as the operator norm of $\hat L$ is less than $1$. Thus, we have 
$$
\langle h, g \rangle_{\H_k} = \langle I_{\nu}, g \rangle_{\H_k} \,\, \forall g \in \H_k. 
$$ 
This implies that $h = I_{\nu} \in \H_{\P_{_{\leq1}}}$, which completes the proof. 
\end{proof}

\begin{lemma}\label{Lq-closed}
$L_q\bigl(\mu,\H_{\P_{_{\leq1}}}\bigr)$ is closed in $L_q\bigl(\mu,\H_k\bigr)$ with respect to the weak kernel mean embedding topology when $q \in (1,\infty)$. 
\end{lemma}

\begin{proof}
Note that \(L_q\bigl(\mu,\H_{\P_{_{\leq1}}}\bigr)\) is a convex subset of \(L_q\bigl(\mu,\H_{k}\bigr)\). Since \(q \in (1,\infty)\), the space \(L_q\bigl(\mu,\H_{k}\bigr)\) is reflexive. Therefore, the weak\(^*\)-topology (i.e., the weak kernel mean embedding topology) coincides with the weak topology. Additionally, for convex sets, closedness with respect to the weak topology and the norm topology are equivalent. Thus, it suffices to show that \(L_q\bigl(\mu,\H_{\P_{_{\leq1}}}\bigr)\) is closed in the norm topology.

Let $\gamma_{\lambda} \rightarrow \gamma$ with respect to the norm-topology; i.e., 
$$
\int_{\sY} \|\gamma_{\lambda}(y) - \gamma(y) \|_{\M}^q \, \mu(dy) \rightarrow 0,
$$
where $\gamma_{\lambda} \in L_q\bigl(\mu,\H_{\P_{_{\leq1}}}\bigr)$ for all $\lambda$. Since convergence in mean implies almost sure convergence along a sub-sequence, there exists a sub-sequence $\{\gamma_{\theta}\}$ of $\{\gamma_{\lambda}\}$ such that for $\mu$-almost everywhere, we have 
$$
\|\gamma_{\theta}(y) - \gamma(y) \|_{\M} \rightarrow 0.
$$
Fix any $y$ that satisfies above convergence result, where these $y$ values have probability $1$. We complete the proof by following the same idea that was used in the proof of Lemma~\ref{subprobability-closed}. Let us define the functional $L^y: \H_k \rightarrow \R$ as follows $L^y(g) \triangleq \langle \gamma(y), g \rangle_{\H_k}$. This functional is linear. Moreover, by above $\mu$-almost everywhere convergence result and the fact that 
$$
\sup_{\theta} \sup_{g \in \H_k: \|g\| \leq 1} |\langle \gamma_{\theta}(y), g \rangle_{\H_k}| \leq 1,
$$
we also have
\begin{align*}
\sup_{g \in \H_k: \|g\| \leq 1} |L^y(g)| \leq 1. 
\end{align*}
Hence, $L^y$ is a bounded linear functional on $\H_k$, where $\H_k$ is endowed with sup-norm. Since $\H_k$ is dense in $C_0(\sU)$, by continuous linear extension theorem, we can uniquely extend $L^y$ to $C_0(\sU)$ as a bounded linear functional, denoted as $\hat L^y$. But since $C_0(\sU)^* = \M(\sU)$, there exists $\nu^y \in \M(\sU)$ such that $\hat L^y = \nu^y$; that is, for any $g \in C_0(\sU)$, we have $\hat L^y(g) = \int_{\sU} g(u) \, \nu^y(du)$. One can prove that $\hat L^y$ is a positive functional; that is, if $g \geq 0$, then $\hat L^y(g) \geq 0$. Hence, $\nu^y \in P_{_{\leq1}}$ as the operator norm of $\hat L^y$ is less than $1$. Hence, we have 
$$
\langle \gamma(y), g \rangle_{\H_k} = \langle I_{\nu^y}, g \rangle_{\H_k} \,\, \forall g \in \H_k. 
$$ 
This implies that $\gamma(y) = I_{\nu^y} \in \H_{\P_{_{\leq1}}}$ for $\mu$-almost every $y$. This completes the proof.
\end{proof}

Since $L_q\bigl(\mu,\H_{\P_{_{\leq1}}}\bigr)$ is closed and bounded in $L_q\bigl(\mu,\H_k\bigr)$ with respect to weak kernel mean embedding topology, it is (sequentially) compact and metrizable with respect to this topology by Banach-Alaoglu Theorem. Note that $L_q\bigl(\mu,\H_{\P}\bigr)$ is a subset of $L_q\bigl(\mu,\H_{\P_{_{\leq1}}}\bigr)$, and so, we can endow $L_q\bigl(\mu,\H_{\P}\bigr)$ with relative weak kernel mean embedding topology. 

\begin{tcolorbox} 
[colback=white!100]
\begin{definition}[Weak Kernel Mean Embedding Topology]\label{defnWeakKMET}
Weak kernel mean embedding topology on $$L_q(\mu,\H_{\P}) = I \circ \Gamma \, (\equiv \Gamma)$$ is the relative weak kernel mean embedding topology; that is, $\gamma_{\lambda} \rightharpoonup^* \gamma$ in $L_q(\mu,\H_{\P})$, if
\begin{align}
&\int_{\sY} \langle \gamma_{\lambda}(y), f(y) \rangle_{\H_k} \, \mu(dy) \rightarrow \int_{\sY} \langle \gamma(y), f(y) \rangle_{\H_k} \, \mu(dy), \nonumber
\end{align}
for all $f \in L_p\bigl(\mu,\H_k\bigr)$. 
\end{definition}
\end{tcolorbox}

\begin{remark}
In the definition above, the conjugate pair \(p\) and \(q\) would typically be explicitly stated. However, as noted in Remark~\ref{equivalence-weakstar}, the choice of \(p\) and \(q\) does not affect the definition. For \(1 \leq p < \infty\), the weak kernel mean embedding topologies induced by $L_p\bigl(\mu,\H_k\bigr)$ are equivalent.
\end{remark}

As we will see soon, similar to the $w^*$-topology, the set of stochastic kernels is not closed under this topology unless $\sU$ is compact as shown below. 

\begin{lemma}\label{nonclosed-MMD}
$L_q\bigl(\mu,\H_{\P}\bigr)$ is not closed with respect to the weak kernel mean embedding topology unless $\sU$ is compact.
\end{lemma}

\begin{proof}
First note that if $\sU$ is compact, then $C_0(\sU) = C_b(\sU) = C(\sU)$. Let $\gamma_{\lambda} \rightharpoonup^* \gamma$ with respect to weak kernel mean embedding topology. Since $L_q\bigl(\mu,\H_{\P_{_{\leq1}}}\bigr)$ is closed in $L_q\bigl(\mu,\H_k\bigr)$ with respect to weak kernel mean embedding topology, we have $\gamma \in L_q\bigl(\mu,\H_{\P_{_{\leq1}}}\bigr)$. We now prove that $\gamma$ is indeed an element of $L_q\bigl(\mu,\H_{\P}\bigr)$. By $\gamma_{\lambda} \rightharpoonup^* \gamma$, we have $\mu \otimes \gamma_{\lambda}(\sY \times \cdot) \rightarrow \mu \otimes \gamma(\sY \times \cdot)$ in weak MMD topology\footnote{A sequence $\{\nu_n\} \in \M(\sU)$ converges to $\nu$ in weak MMD topology if 
$$
\langle \I_{\nu_n}, h \rangle_{\H_k} = \int_{\sU} h(u) \, \nu_n(du) \rightarrow  \langle \I_{\nu}, h \rangle_{\H_k} = \int_{\sU} h(u) \, \nu(du) \,\, \forall h \in \H_k. 
$$
The relationship between the MMD topology and the weak MMD topology is very similar to the relationship between the topology induced by the total variation norm and the weak convergence topology. In the former case, the convergence of the integral is uniform over the unit ball in the function class of interest, while in the latter case, the convergence is pointwise within that function class.}; that is, 
$$
\int_{\sY} \int_{\sU} h(y) \, \gamma_{\lambda}(y)(du) \, \mu(dy) \rightarrow \int_{\sY} \int_{\sU} h(y) \, \gamma(y)(du) \, \mu(dy) \,\, \forall h \in \H_k.
$$
By \cite[Lemma 5 and Remark 6]{SiBaScMa24}, $\mu \otimes \gamma_{\lambda}(\sY \times \cdot) \rightarrow \mu \otimes \gamma(\sY \times \cdot)$ in vague topology, and so, in weak convergence topology as $C_0(\sU) = C_b(\sU) = C(\sU)$. Since $\mu \otimes \gamma_{\lambda}(\sY \times \cdot) \in \P(\sU)$ for all $\lambda$ and $\P(\sU)$ is a closed subset of $\M(\sU)$ under weak convergence topology, we have $\mu \otimes \gamma(\sY \times \cdot) \in \P(\sU)$. Hence, $\gamma \in \rL_q\bigl(\mu,\P(\sU)\bigr)$.

Suppose now that $\sU$ is not compact. Using the same construction as in the proof of Lemma~\ref{nonclosed}, we can find a sequence of probability measures $\{\nu_n\}$ such that $\nu_n$ converges to $0$ vaguely, where $0(\,\cdot\,)$ denotes the degenerate measure on $\sU$. Then, define $\gamma_n(y)(\,\cdot\,) = \nu_n(\,\cdot\,)$ and $\gamma(y)(\,\cdot\,) = 0(\,\cdot\,)$. Let $f \in L_p\bigl(\mu,\H_k\bigr)$. Then we have
\begin{align}
&\lim_{n\rightarrow\infty} \int_{\sY} \langle \gamma_n(y), f(y) \rangle \, \mu(dy) 
= \lim_{n\rightarrow\infty} \int_{\sY} \int_{\sU} f(y)(u) \, \nu_n(du) \, \mu(dy) \nonumber \\
&= \int_{\sY} \lim_{n\rightarrow\infty} \int_{\sU} f(y)(u) \, \nu_n(du) \, \mu(dy) \text{ (as $\|f(y)\|_{\H_k}$ is $\mu$-integrable)} \nonumber \\
&= 0 \text{ (as $f(y) \in \H_k \subset C_0(\sU)$)}. \nonumber
\end{align}
Hence, $\gamma_n \rightharpoonup^* \gamma$, where $\gamma(y) = 0$ for all $y$. But, $\gamma \notin L_q\bigl(\mu,\P(\sU)\bigr)$, and so, $L_q\bigl(\mu,\P(\sU)\bigr)$ is not closed in $L_q\bigl(\mu,\P_{_{\leq1}}(\sU)\bigr)$.
\end{proof}

Recall that $L_q\bigl(\mu,\H_{\mathcal{P}_{\leq 1}}\bigr)$ is (sequentially) compact with respect to the weak kernel mean embedding topology by the Banach–Alaoglu Theorem. Since $L_q\bigl(\mu,\H_{\mathcal{P}}\bigr)$ is a subset of $L_q\bigl(\mu,\H_{\mathcal{P}_{\leq 1}}\bigr)$, it is \emph{relatively} (sequentially) compact under the same topology. Moreover, when $\sU$ is compact, this set is also (sequentially) compact by Lemma~\ref{nonclosed-MMD}. In contrast, when $\sU$ is non-compact, Lemma~\ref{nonclosed-MMD} shows that $L_q\bigl(\mu,\H_{\mathcal{P}}\bigr)$ is no longer closed, and therefore compactness is not guaranteed.

\begin{remark}

In the kernel mean embedding topology, the conjugate pair \((p, q) = (2, 2)\) is particularly useful due to its compatibility with the Hilbert space structure.  
More concretely, we work with the Hilbert space \(L_2(\mu, \H_k)\) in this case. The inner product on \(L_2(\mu, \H_k)\) is given by:
\[
\langle\langle \gamma, f \rangle\rangle = \int_{\sY} \langle \gamma(y), f(y) \rangle_{\H_k} \, \mu(dy).
\]
By Theorem~\ref{mainduality}, we know that the topological dual of \(L_2(\mu, \H_k)\) is isometrically isomorphic to itself:
\[
L_2(\mu, \H_k)^* = L_2(\mu, \H_k).
\]
This duality is expressed through the inner product defined above. Consequently, we can equip \(L_2(\mu, \H_k)\) with the weak\(^*\)-topology (i.e., weak kernel mean embedding topology), which is induced by \(L_2(\mu, \H_k)\). 

Under this topology, convergence \(\gamma_{\lambda} \rightharpoonup^* \gamma\) in \(L_2(\mu, \H_k)\) occurs if and only if:
\[
\int_{\sY} \langle \gamma_{\lambda}(y), f(y) \rangle_{\H_k} \, \mu(dy) \rightarrow \int_{\sY} \langle \gamma(y), f(y) \rangle_{\H_k} \, \mu(dy),
\]
for all \(f \in L_2(\mu, \H_k)\). As shown above, while the space \(L_2(\mu, \H_{\P})\) is relatively (sequentially) compact under this topology, it is not closed unless \(\sU\) is compact. Finally, in this particular case, because of the Hilbert space structure, this topology also corresponds to the weak topology on $L_2\bigl(\mu,\H_k\bigr)$. 
\end{remark}

The following result establishes that convergence under the weak kernel mean embedding topology remains valid even if the reference probability measure \(\mu\) on \(\sY\) is replaced by another probability measure \(\eta\), provided \(\eta\) is absolutely continuous with respect to \(\mu\). A similar result is established in \cite[Lemma 3.6]{yuksel2023borkar} for Young narrow topology.

\begin{lemma}\label{equivKernelT}
Let \(\eta \ll \mu\), and suppose the Radon-Nikodym derivative of \(\eta\) with respect to \(\mu\) is \(l\)-integrable with respect to \(\mu\), where $l > 1$. Then, \(\gamma_n \rightharpoonup^* \gamma\) under the weak kernel mean embedding topology at input \(\mu\) implies \(\gamma_n \rightharpoonup^* \gamma\) under the weak kernel mean embedding topology at input \(\eta\), provided \(p\) is replaced by \(\tilde{p}\) satisfying \(1/\tilde{p} + 1/l = 1/p\). Indeed, since the weak kernel mean embedding topologies are equivalent for any \(1 \leq p < \infty\), there is no need to adjust \(p\).
\end{lemma}

\begin{proof}
By assumption we have that there exists $f: \mathbb{X} \to \mathbb{R}_+$ with
\[f(x) = \frac{d \eta}{d \mu}(x),\]
which is the Radon-Nikodym derivative. It is assumed that $f$ is $l$-integrable under $\mu$. Let $g \in L_{\tilde p}(\eta,\H_k)$. Then we define $\tilde g(y)(\cdot) \triangleq f(y) \, g(y)(\cdot)$. By generalized H\"{o}lder's inequality, we have 
$$
\bigg|\int_{\sY} \|\tilde g(y)\|_{\H_k}^{p} \mu(dy)\bigg|^{1/p} \leq  \|f\|_l \, \|g\|_{\tilde p} < \infty.  
$$
Hence $\tilde g \in L_p(\mu,\H_k)$. This completes the proof. 
\end{proof}

\subsection{Two weak topologies revisited}

In this section, we present two commonly studied weak topologies on stochastic kernels: The Borkar ($w^*$) topology and the Young topology on kernels. While this section serves as a review, we will also obtain some useful properties that will be utilized later in the paper.  

\subsubsection{$w^*$-topology (or Borkar topology) on Stochastic Kernels}\label{weak-star}

In this section, we introduce the $w^*$-topology on the set of stochastic kernels, a concept developed in \cite{borkar1989topology}\cite[Section 2.4]{BoArGh12} for randomized Markov policies. Later, in \cite{Sal20},  $w^*$-topology is used to establish the existence of optimal policies in team decision problems. This topology has proven to be a valuable tool in analyzing such problems. We will begin by recalling its formal definition and then proceed to discuss some of its fundamental properties.  

Let $\sU$ be a locally compact Borel space endowed with its Borel $\sigma$-algebra $\U$. For any $g \in C_0(\sU)$, let
\begin{align}
\|g\| \triangleq \sup_{u \in \sU} |g(u)| \nonumber
\end{align}
which turns $(C_0(\sU),\|\cdot\|)$ into a separable Banach space. Let $\|\cdot\|_{TV}$ denote the total variation norm on $\M(\sU)$. The following result is a well-known Riesz-Markov-Kakutani representation theorem, which states that the topological dual of \( C_0(\sU) \) is \( \M(\sU) \).

\begin{theorem}\cite[Theorem 7.17]{Fol99}
For any $\nu \in \M(\sU)$ and $g \in C_0(\sU)$, let $I_{\nu}(g) \triangleq \langle \nu , g \rangle$, where
\begin{align}
\langle \nu, g \rangle \triangleq \int_{\sU} g(u) \, \nu(du). \nonumber
\end{align}
Then the map $\nu \mapsto I_{\nu}$ is an isometric isomorphism from $\M(\sU)$ to $C_0(\sU)^*$. Hence, we can identify $C_0(\sU)^*$ with $\M(\sU)$. For any $g \in C_0(\sU)$ and $\nu \in \M(\sU)$, the duality pairing is given by
\begin{align}
\langle \nu, g \rangle = \int_{\sU} g(u) \, \nu(du).\nonumber
\end{align}
Furthermore, the norm on $\M(\sU)$ induced by this duality is the total variation norm.
\end{theorem}

In this particular case, we have the following identifications \( \sB = C_0(\sU) \) and $\sB^* = \M(\sU)$. Note that, in this case, $\sB^* = \M(\sU)$ is not reflexive. Therefore, we must take into account weakly$^*$-measurable functions when considering duality results. In view of this, we can define the following Banach spaces $L_p\bigl(\mu,C_0(\sU)\bigr)$ and $\rL_q\bigl(\mu,\M(\sU)\bigr)$ with the following norms, respectively\footnote{In \cite{Sal20}, we take $p = 1$ and $q = \infty$ to prove the existence of optimal policies in team decision problems. In this case, the weak$^*$-topology corresponding to this conjugate pair has a nice connection with the Young narrow topology on stochastic kernels. Here, we generalize that topology to arbitrary conjugate pairs $(p, q)$.},
\begin{align}
\|f\|_p &= \left( \int_{\sY} \|f(y)\|^p \, \mu(dy) \right)^{1/p}, \label{eq3} \\
\|\gamma\|_q &= \begin{cases}
\left( \int_{\sY} \|\gamma(y)\|^q_{TV} \, \mu(dy) \right)^{1/q}, & \,\, \text{if} \,\, q < \infty \\
\ess \sup_{y \in \sY} \|\gamma(y)\|_{TV}, & \,\, \text{if} \,\, q = \infty.
\end{cases}
\label{eq4}
\end{align}
Since $C_0(\sU)$ is separable, $L_p\bigl(\mu,C_0(\sU)\bigr)$ is also separable \sy{for $1\leq p < \infty$}. By Theorem~\ref{mainduality}, the topological dual of $L_p\bigl(\mu,C_0(\sU)\bigr)$ is $\rL_q\bigl(\mu,\M(\sU)\bigr)$; that is
\begin{align}
L_p\bigl(\mu,C_0(\sU)\bigr)^* = \rL_q\bigl(\mu,\M(\sU)\bigr). \nonumber
\end{align}
For any $f \in L_p\bigl(\mu,C_0(\sU)\bigr)$ and $\gamma \in \rL_q\bigl(\mu,\M(\sU)\bigr)$, the duality pairing is given by
\begin{align}
\langle\langle \gamma, f \rangle\rangle &= \int_{\sY} \langle \gamma(y), f(y) \rangle \, \mu(dy) \nonumber \\
&= \int_{\sY} \int_{\sU} f(y)(u) \, \gamma(y)(du) \, \mu(dy). \label{eq5}
\end{align}
Hence, we can equip $\rL_q\bigl(\mu,\M(\sU)\bigr)$ with weak$^*$ topology induced by $L_p\bigl(\mu,C_0(\sU)\bigr)$. We refer to this topology as $w^*$-topology. Under this topology, $\gamma_{\lambda} \rightharpoonup^* \gamma$ in $\rL_q\bigl(\mu,\M(\sU)\bigr)$, if
\begin{align}
&\int_{\sY} \int_{\sU} f(y)(u)\, \gamma_{\lambda}(y)(du)\, \mu(dy) \rightarrow \int_{\sY} \int_{\sU} f(y)(u)\, \gamma(y)(du)\, \mu(dy), \nonumber
\end{align}
for all $f \in L_p\bigl(\mu,C_0(\sU)\bigr)$.

\begin{tcolorbox} 
[colback=white!100]
\begin{definition}[Stochastic Kernel] \cite[Definition C.1]{HeLa96}
A mapping $\gamma: \sY \rightarrow \M(\sU)$ is called a \emph{stochastic kernel} from $\sY$ to $\sU$ if, for all $D \in \U$, the mapping $\sY \ni y \mapsto \gamma(y)(D) \in \R$ is $\Y / \B(\R)$-measurable and $\gamma(y) \in \P(\sU)$ for all $y \in \sY$. Let $\Gamma$ denote the set of all stochastic kernels from $\sY$ to $\sU$.
\end{definition}
\end{tcolorbox}

The following result offers an alternative characterization of stochastic kernels.

\begin{lemma}\label{kernel}
The set of stochastic kernels $\Gamma$ is the same as the weakly$^*$ measurable functions from $\sY$ to $\M(\sU)$ whose range is contained in $\P(\sU)$. 
\end{lemma}

\begin{proof}
Let \(\gamma\) be a weakly$^*$ measurable function from \(\sY\) to \(\M(\sU)\) whose range is contained in \(\P(\sU)\).

To begin, we note that for any continuous and bounded function \(g\) on \(\sU\), the mapping \(x \mapsto \langle \gamma(x), g \rangle \in \mathbb{R}\) is \(\mathcal{X} / \mathcal{B}(\mathbb{R})\)-measurable. This is due to the fact that any such \(g\) can be approximated pointwise by a sequence \(\{g_n\}_{n \geq 1} \subset C_0(\sU)\), where \(\|g_n\| \leq \|g\|\) for all \(n\).

Next, consider any closed set \(F \subset \sU\). The indicator function \(1_F\) can be approximated pointwise by the continuous and bounded functions \(h_n(u) = \max(1 - n d_{\sU}(u, F), 0)\), where \(d_{\sU}\) denotes the metric on \(\sU\) and \(d_{\sU}(u, F) = \inf_{y \in F} d_{\sU}(u, y)\). Consequently, this implies that the mapping \(x \mapsto \gamma(x)(F) \in \mathbb{R}\) is also \(\mathcal{X} / \mathcal{B}(\mathbb{R})\)-measurable for any closed set \(F\) in \(\sU\). Thus, we can conclude that $\gamma$ is a stochastic kernel by \cite[Proposition 7.25]{BeSh78}.

The reverse implication is relatively straightforward.
\end{proof}

In view of Lemma~\ref{kernel}, we can identify the set of stochastic kernels $\Gamma$ via $\rL_q(\mu,\P(\sU))$; that is,
\begin{tcolorbox} 
[colback=white!100]
$$\rL_q(\mu,\P(\sU)) = \Gamma.$$
\end{tcolorbox}  
Let $\P_{_{\leq1}}(\sU)$ denote the set of sub-probability measures in $\M(\sU)$. Then, since $\rL_q\bigl(\mu,\P_{_{\leq1}}(\sU)\bigr)$ is proved to be closed in $\rL_q\bigl(\mu,\M(\sU)\bigr)$ with respect to $w^*$-topology (see \cite{Sal20}) and is a bounded subset, it is (sequentially) compact and metrizable with respect to $w^*$-topology by Banach-Alaoglu Theorem.  
Note that $\rL_q\bigl(\mu,\P(\sU)\bigr)$ is a subset of $\rL_q\bigl(\mu,\P_{_{\leq1}}(\sU)\bigr)$, and so, we can endow $\rL_q\bigl(\mu,\P(\sU)\bigr)$ with relative $w^*$-topology. Unfortunately, $\rL_q\bigl(\mu,\P(\sU)\bigr)$ is not closed in this topology as shown below\footnote{In \cite{Sal20}, we established this result by constructing a counterexample specifically for the case when \(\sU = \mathbb{R}\), which is relatively easy. In Lemma~\ref{nonclosed} above, we extend this by providing a counterexample applicable to any locally compact space \(\sU\). This counterexample draws inspiration from the elegant proof presented in \cite[Lemma 12]{SiBaScMa24}.}.

\begin{remark}
Similar to the weak kernel mean embedding topology, the conjugate pair \(p\) and \(q\) need not be explicitly stated in the definition of the \(w^*\)-topology above. As noted in Remark~\ref{equivalence-weakstar}, the choice of \(p\) and \(q\) does not affect the definition. For \(1 \leq p < \infty\), the \(w^*\)-topologies are equivalent.
\end{remark}

\begin{lemma}\label{nonclosed}
$\rL_q\bigl(\mu,\P(\sU)\bigr)$ is not closed with respect to $w^*$-topology unless $\sU$ is compact.
\end{lemma}

\begin{proof}
First note that if $\sU$ is compact, then $C_0(\sU) = C_b(\sU) = C(\sU)$. Moreover, if $\gamma_{\lambda} \rightharpoonup^* \gamma$ with respect to $w^*$-topology, then $\mu \otimes \gamma_{\lambda}(\sY \times \cdot) \rightarrow \mu \otimes \gamma(\sY \times \cdot)$ in vague topology, and so, in weak convergence topology as $C_0(\sU) = C_b(\sU) = C(\sU)$. Since $\mu \otimes \gamma_{\lambda}(\sY \times \cdot) \in \P(\sU)$ for all $\lambda$ and $\P(\sU)$ is a closed subset of $\M(\sU)$ under weak convergence topology, we have $\mu \otimes \gamma(\sY \times \cdot) \in \P(\sU)$. Hence, $\gamma \in \rL_q\bigl(\mu,\P(\sU)\bigr)$. 

Suppose now that $\sU$ is not compact. Since $C_0(\sU)$ is separable, there exists a countable subset $\{g_n\}$ of $C_0(\sU)$ that is dense in $C_0(\sU)$ and determines the vague topology; that is, 
$$
\nu_{\lambda} \rightarrow \nu \,\, \text{vaguely} \,\, \iff \int_{\sU} g_n(u) \, \nu_{\lambda}(du) \rightarrow \int_{\sU} g_n(u) \, \nu(du) \,\, \forall n. 
$$
For any positive integer $n$, we now construct a sequence $\{u_1,\ldots,u_n\}$ such that $|g_j(u_i)| \leq 1/n$ for all $j \leq i$. Indeed, fix any $n$. Since $g_1 \in C_0(\sU)$, there exists a compact set $K_1$ such that $\sup_{u \in K_1^c} |g_1(u)| < 1/n$. Pick $u_1$ from $K_1^c$. Again since $g_2 \in C_0(\sU)$, there exists a compact set $K_2$ such that $K_1 \subset K_2$ and $\sup_{u \in K_2^c} |g_2(u)| < 1/n$. Note that since $K_1 \subset K_2$, we have also $\sup_{u \in K_2^c} |g_1(u)| < 1/n$. Pick $u_2$ from $K_2^c$. Continue this procedure until $u_n$. This sequence satisfies the above condition. For this sequence $\{u_1,\ldots,u_n\}$, we construct the following probability measure: 
$$
\nu_n(\cdot) \triangleq \frac{1}{n} \sum_{i=1}^n \delta_{u_i}(\cdot). 
$$
Let $g \in C_0(\sU)$. Fix any $\varepsilon > 0$. Then, there exists $g_j$ such that $\|g_j-g\| < \varepsilon$. Hence, we have 
\begin{align*}
\left|\int_{\sU} g(u) \, \nu_n(du) \right|  &= \left|\frac{1}{n} \sum_{i=1}^n  g(u_i) \right| \leq \left|\frac{1}{n} \sum_{i=1}^n  g(u_i)-g_j(u_i) \right| + \left|\frac{1}{n} \sum_{i=1}^n  g_j(u_i) \right| \\
&\leq \varepsilon + \frac{1}{n} \sum_{i=1}^{j-1}  \|g_j\| + \frac{1}{n} \sum_{i=j}^n  \frac{1}{n} \rightarrow \varepsilon \,\, \text{as} \,\, n \rightarrow \infty. 
\end{align*}
Since $\varepsilon$ is arbitrary, we have $\int_{\sU} g(u) \, \nu_n(du)  \rightarrow 0$ as $n \rightarrow \infty$ for any $g \in C_0(\sU)$. Hence, $\nu_n$ converges to $0$, where $0(\,\cdot\,)$ denotes the degenerate measure on $\sU$; that is, $0(D)=0$ for all $D \in \U$. 

Using these probability measures, we define $\gamma_n(y)(\,\cdot\,) = \nu_n(\,\cdot\,)$ and $\gamma(y)(\,\cdot\,) = 0(\,\cdot\,)$. Let $f \in L_p\bigl(\mu,C_0(\sU)\bigr)$. Then we have
\begin{align}
&\lim_{n\rightarrow\infty} \int_{\sY} \langle \gamma_n(y), f(y) \rangle \, \mu(dy) 
= \lim_{n\rightarrow\infty} \int_{\sY} \int_{\sU} f(y)(u) \, \nu_n(du) \, \mu(dy) \nonumber \\
&= \int_{\sY} \lim_{n\rightarrow\infty} \int_{\sU} f(y)(u) \, \nu_n(du) \, \mu(dy) \text{ (as $\|f(y)\|$ is $\mu$-integrable)} \nonumber \\
&= 0 \text{ (as $f(y) \in C_0(\sU)$)}. \nonumber
\end{align}
Hence, $\gamma_n \rightharpoonup^* \gamma$, where $\gamma(y) = 0$ for all $y$. But, $\gamma \notin \rL_q\bigl(\mu,\P(\sU)\bigr)$, and so, $\rL_q\bigl(\mu,\P(\sU)\bigr)$ is not closed in $\rL_q\bigl(\mu,\P_{_{\leq1}}(\sU)\bigr)$.
\end{proof}

In view of Lemma~\ref{nonclosed}, $\rL_q\bigl(\mu,\P(\sU)\bigr)$ is relatively (sequentially) compact with respect to $w^*$-topology. 

\subsubsection{The Young Narrow Topology on Stochastic Kernels}

Another widely used weak topology for stochastic kernels in control theory is the Young narrow topology \cite{Val94}. In this framework, stochastic kernels are represented as probability measures defined on a product space while maintaining a fixed marginal in the input space. Consequently, we associate with $\Gamma$ the set of probability measures induced on the product space $\sY \times \sU$, with a fixed marginal $\mu$ on $\sY$. On this set of probability measures over the product space, the weak convergence topology is defined as the Young narrow topology. Alternatively, this topology can also be defined by viewing stochastic kernels as mappings from $\sY$ to $\P(\sU)$ and employing duality with the set of $C_b$-Caratheodory function valued mappings, which is similar to the construction in $w^*$-topology and kernel mean embedding topology. This perspective is the one we adopt in this section. Before introducing this topology, we first establish the equivalence between strongly measurable functions from $\sY$ to $C_0(\sU)$ and the set of $C_0$-Caratheodory functions.

Note that a measurable function $h:\sY\times\sU\rightarrow\R$ is called a Caratheodory function if it is continuous in $u$ for all $y \in \sY$. Let $\H_k$-Caratheodory functions, denoted as $\Car_k(\sY\times\sU)$, be the set of all Caratheodory functions $h$ such that $h(y,\,\cdot\,) \in \H_k$ for all $y \in \sY$, let $C_0$-Caratheodory functions, denoted as $\Car_0(\sY\times\sU)$, be the set of all Caratheodory functions $h$ such that $h(y,\,\cdot\,) \in C_0(\sU)$ for all $y \in \sY$ and let $C_b$-Caratheodory functions, denoted as $\Car_b(\sY\times\sU)$, be the set of all Caratheodory functions $h$ such that $h(y,\,\cdot\,) \in C_b(\sU)$ for all $y \in \sY$.

\begin{lemma}\label{Car}
Let $f:\sX \rightarrow C_0(\sU)$ be a strongly measurable function and define $h_f:\sX\times\sU\rightarrow \R$ as $h_f(x,u) \triangleq f(x)(u)$. Then, $h_f \in \Car_0(\sX\times\sU)$. Conversely, let $h \in \Car_0(\sX\times\sU)$ and define $f_h:\sX\rightarrow C_0(\sU)$ as $f_h(x) = h(x,\,\cdot\,)$. Then, $f_h$ is strongly measurable. 
\end{lemma}

\begin{proof}
It is straightforward to prove that the forward statement is true for any simple function $f:\sX\rightarrow C_0(\sU)$.  Since any strongly measurable function can be approximated via simple functions by definition, the forward statement is also true for any strongly measurable function. 

Conversely, let $h\in\Car_0(\sX\times\sU)$ and define $f_h:\sX\rightarrow C_0(\sU)$ as $f_h(x) = h(x,\,\cdot\,)$. Then, for any $\gamma \in \M(\sU)$, the function $\sX \ni x \mapsto \int_{\sU} f_h(x)(u) \, \gamma(du) \in \R$ is measurable. Hence, $f_h$ is weakly measurable. But since $C_0(\sU)$ is separable, $f_h$ is also strongly measurable by \cite[Theorem 4.2.4-(c)]{PaWi18}.  
\end{proof}

Therefore, any element $f$ of $L_p\bigl(\mu,C_0(\sU)\bigr)$ is an element of $\Car_0(\sX\times\sU)$ satisfying 
\begin{align}\label{cccv}
\int_{\sY} \|f(y)\|^p \, \mu(dx) < \infty.
\end{align}
The converse is also true; that is, if $h \in \Car_0(\sX\times\sU)$ satisfying (\ref{cccv}), then $h$ is an  element of $L_p\bigl(\mu,C_0(\sU)\bigr)$. We note that a similar conclusion can be made for $\H_k$-Caratheodory and $C_b$-Caratheodory functions. 

Let us now define \emph{Young narrow topology} on $\rL_q\bigl(\mu,\M(\sU)\bigr)$ \cite[Definition 4.7.11]{PaWi18}\footnote{Normally, Young narrow topology is defined only on the set of stochastic kernels $\rL_q\bigl(\mu,\P(\sU)\bigr)$. But, here, we extend this definition to the set of all weakly$^*$-measurable functions from $\sY$ to $\M(\sU)$. Moreover, in the original formulation, the topology is induced by the set of $C_b$-Caratheodory functions that are in $L_1(\mu,C_b(\sU))$. Here, we extend this definition to $L_p(\mu,C_b(\sU))$ for any $1 \leq p < \infty$.}. Young narrow topology on $\rL_q\bigl(\mu,\M(\sU)\bigr)$ is the smallest topology for which the mapping
\begin{align}
\rL_q\bigl(\mu,\M(\sU)\bigr) \ni \gamma \mapsto \int_{\sY} \int_{\sU} f(y)(u) \, \gamma(y)(du) \, \mu(dy) \in \R, \nonumber
\end{align}
is continuous for all $f \in L_p(\mu,C_b(\sU))$. Unfortunately, since $L_p(\mu,C_b(\sU))^* \neq \rL_q\bigl(\mu,\M(\sU)\bigr)$ unless $\sU$ is compact, this is not a weak$^*$-topology. Hence, the relation between $w^*$-topology and Young narrow topology is similar to the vague topology and weak topology on finite signed measures. 

In Young narrow topology, one can prove that $\rL_q\bigl(\mu,\P(\sU)\bigr)$ is a closed (as opposed to $w^*$-topology and weak kernel mean embedding topology) and metrizable subset of $\rL_q\bigl(\mu,\M(\sU)\bigr)$ \cite[Proposition 4.7.14]{PaWi18}. However, $\rL_q\bigl(\mu,\P(\sU)\bigr)$  is not relatively (sequentially) compact in this topology unless $\sU$ is compact because of the lack of the duality result (i.e., The Banach-Alaoglu theorem is not applicable in this context). 

\begin{remark}
Although Young narrow topology is not a weak$^*$-topology, it is still true that the choice of \(p\) and \(q\) does not impact the definition of the topology above. For \(1 \leq p < \infty\), the Young narrow topologies are equivalent.
\end{remark}

\subsection{Comparison of Weak Topologies}

Based on the results established earlier, we observe some intriguing characteristics of the three weak topologies in question. First, both the \( w^* \)-topology and the weak  kernel mean embedding topology are relatively (sequentially) compact, which is a vital property for demonstrating the existence of optimal solutions in various contexts. However, it is important to note that neither of these topologies is closed. Conversely, the Young narrow topology possesses the advantage of being closed but lacks relative (sequential) compactness.

Interestingly, to achieve closure in the \( w^* \)-topology and the weak kernel mean embedding topology, a tightness condition is required. This condition effectively ensures uniform concentration of probability measures on compact sets, a necessity highlighted by counterexample in Lemma~\ref{nonclosed}. Similarly, for the Young narrow topology to attain relative (sequential) compactness, the same tightness condition must be satisfied, as implied by Prokhorov's theorem. Consequently, in terms of ensuring the existence of optimal solutions, all three topologies provide comparable benefits.

Despite these similarities, the \( w^* \)-topology presents a notable advantage due to its well-defined duality structure, which can be effectively utilized in various analyses. In addition to this duality, the weak kernel mean embedding topology exhibits a Hilbert space structure when the conjugate pair is set to \( (p,q) = (2,2) \). .

It is interesting to note that the \( w^* \)-topology, weak kernel mean embedding topology, and Young narrow topology are all topologically equivalent on the set of stochastic kernels \( \Gamma \). However, since completeness is not classified as a topological property, this equivalence does not contradict the distinctions outlined above.

\begin{theorem}\label{equivalence}
Let $\{\gamma_{\lambda}\}$ be a sequence in $\Gamma$ and let $\gamma$ be also in $\Gamma$. Then, the following are equivalent:
\begin{itemize}
\item[(i)] $\gamma_{\lambda} \rightharpoonup^* \gamma$ with respect to weak kernel mean embedding topology.
\item[(ii)] $\gamma_{\lambda} \rightharpoonup^* \gamma$ with respect to $w^*$-topology.
\item[(ii)] $\gamma_{\lambda} \rightharpoonup^* \gamma$ with respect to Young narrow topology.
\end{itemize}
\end{theorem}

\begin{proof}
We use the following equivalent characterizations of these topologies to establish the result:
\begin{itemize}
\item[(i)] $\gamma_{\lambda} \rightharpoonup^* \gamma$ with respect to weak kernel mean embedding topology if and only if for all $E \in \Y$, $\mu \otimes \gamma_{\lambda}(E \times \cdot) \rightarrow \mu \otimes \gamma(E \times \cdot)$ in MMD topology. 
\item[(ii)] $\gamma_{\lambda} \rightharpoonup^* \gamma$ with respect to $w^*$-topology if and only if for all $E \in \Y$, $\mu \otimes \gamma_{\lambda}(E \times \cdot) \rightarrow \mu \otimes \gamma(E \times \cdot)$ in vague topology. 
\item[(iii)] $\gamma_{\lambda} \rightharpoonup^* \gamma$ with respect to Young narrow topology if and only if for all $E \in \Y$, $\mu \otimes \gamma_{\lambda}(E \times \cdot) \rightarrow \mu \otimes \gamma(E \times \cdot)$ in weak topology. 
\end{itemize}

The proofs of these results are almost identical, therefore, we only prove (i). To prove the forward direction of (i), let $\gamma_{\lambda} \rightharpoonup^* \gamma$ with respect to weak kernel mean embedding topology. Fix any $E \in \Y$. For any $h \in \H_k$, define the following simple function $f \triangleq h \, 1_E$ in $L_p(\mu,\H_k)$. Then, we have
\begin{align*}
\langle \langle \gamma_{\lambda},f \rangle \rangle = \int_{\sY} \int_{E} h(u) \, \gamma_{\lambda}(y)(du) \, \mu(dy)  \rightarrow \langle \langle \gamma,f \rangle \rangle = \int_{\sY} \int_{E} h(u) \, \gamma(y)(du) \, \mu(dy). 
\end{align*} 
Hence, $\mu \otimes \gamma_{\lambda}(E \times \cdot) \rightarrow \mu \otimes \gamma(E \times \cdot)$ in weak MMD topology. Since $\mu \otimes \gamma_{\lambda}(E \times \sU) = \mu \otimes \gamma(E \times \sU) = \mu(E)$, by \cite[Lemma 5 and Remark 6]{SiBaScMa24}, $\mu \otimes \gamma_{\lambda}(E \times \cdot) \rightarrow \mu \otimes \gamma(E \times \cdot)$ in MMD topology. This completes the proof of the forward part. 

For the converse part, for all $E \in \Y$, let $\mu \otimes \gamma_{\lambda}(E \times \cdot) \rightarrow \mu \otimes \gamma(E \times \cdot)$ in MMD topology. Consider any simple function $f \triangleq \sum_{i=1}^k h_i \, 1_{E_i}$ in $L_p(\mu,\H_k)$. Then, we have 
\small
\begin{align*}
\langle \langle \gamma_{\lambda},f \rangle \rangle = \sum_{i=1}^k \int_{\sY} \int_{E_i} h_i(u) \, \gamma_{\lambda}(y)(du) \, \mu(dy)  \rightarrow \langle \langle \gamma,f \rangle \rangle = \sum_{i=1}^k \int_{\sY} \int_{E_i} h_i(u) \, \gamma(y)(du) \, \mu(dy)
\end{align*}
\normalsize
as $\mu \otimes \gamma_{\lambda}(E_i \times \cdot) \rightarrow \mu \otimes \gamma(E_i \times \cdot)$ in MMD topology (and so in weak MMD topology), for all $i=1,\ldots,k$. Since we can approximate any function in $L_p(\mu,\H_k)$ via a sequence of simple functions pointwise, we can conclude that above convergence is true for any $f \in L_p(\mu,\H_k)$, which concludes the proof of the converse part.  

Now using the above characterizations of the topologies, we complete the proof. 

Note that for a fixed \( E \in \mathcal{Y} \), the total masses of the sub-probability measures \( \{\mu \otimes \gamma_{\lambda}(E \times \cdot)\} \) and \( \mu \otimes \gamma(E \times \cdot) \) are both equal to \( \mu(E) \) (assuming, without loss of generality, that \( \mu(E) > 0 \)). As used several times so far, by \cite[Lemma 5]{SiBaScMa24}, the MMD topology, vague topology, and weak convergence topology are equivalent for the set of probability measures. Since the aforementioned sub-probability measures have the same mass as \( \mu(E) \), we can normalize them using the constant \( \mu(E) \) to convert them into probability measures and apply the equivalence result from \cite[Lemma 5]{SiBaScMa24}. This completes the proof.
\end{proof}

In Theorem~\ref{equivalence}, the most significant assumption used in the proof is that \(\gamma\) is also in \(\Gamma\); that is, \(\gamma(y) \in \P(\sU)\) (or in \(\H_{\P}\)) for almost every \(y\). If we do not assume that \(\gamma \in \Gamma\) and only assume that the sequence \(\{\gamma_{\lambda}\}\) converges in the respective topologies, then the topologies are not equivalent. For instance, in \(w^*\)-topology or weak kernel mean embedding topology, we may converge to sub-probability measures, which is not the case in the Young narrow topology. Therefore, the equivalence of these topologies will no longer hold once we consider finite signed measure-valued kernels: \(\gamma: \sY \rightarrow \M(\sU) \,\, \text{or} \,\, \H_{\M}\).

In this section, we finally present the result that formalizes the tightness discussion outlined above. Before stating the theorem, however, we will first provide a precise mathematical definition of tightness. Let $\G$ be a subset of $\P(\sU)$. It is called \emph{tight} if for any $\varepsilon > 0$, there exists a compact set $K_{\varepsilon} \subset \sU$ such that 
$$
\sup_{\nu \in \G} \nu(\sU \setminus K_{\varepsilon}) < \varepsilon. 
$$
In other words, tightness ensures that distributions in $\G$ uniformly concentrates on compact sets.

\begin{theorem}\label{tightness}
Let $\G$ be a subset of $\P(\sU)$ and let $q \in (1,\infty)$.
\begin{itemize}
\item[(i)] The set $\rL_q(\mu,\G)$ has the closure with respect to the $w^*$-topology in $\Gamma$ if $\G$ is tight. 
\item[(ii)] The set $L_q(\mu,\H_{\G})$ has the closure with respect to the weak kernel mean embedding topology in $L_q(\mu,\H_{\P}) \equiv \Gamma$  if  $\G$ is tight. 
\item[(iii)] The set $\rL_q(\mu,\G)$ is relatively (sequentially) compact with respect to the Young narrow topology if $\G$ is tight.
\end{itemize}
The converses of above statements hold when we additionally assume that $\sU$ is complete (i.e., a Polish space).
\end{theorem}

\begin{proof}
We begin with the proof of item (iii). First, suppose that $\G$ is tight. Then, define the following set:
$$
\C_{\G} \triangleq \{\mu \otimes \gamma: \gamma \in \rL_q(\mu,\G)\} \subset \P(\sY\times\sU),
$$
where $\mu \otimes \gamma(dy,du) \triangleq \mu(dy) \, \gamma(y)(du)$. One can prove that $\C_{\G}$ is tight in the sense of the Definition~4.7.20 in \cite{PaWi18}. Hence, by \cite[Theorem 4.7.22]{PaWi18}, $\rL_q(\mu,\G)$ is relatively (sequentially) compact with respect to the Young narrow topology. Conversely, suppose that $\rL_q(\mu,\G)$ is relatively (sequentially) compact with respect to the Young narrow topology and $\sU$ is complete. This implies that $\G$ is also relatively  (sequentially) compact with respect to the weak topology. Indeed, let $\{\nu_n\}$ be a sequence in $\G$ and for each $n$, define the constant function $\gamma_n(y) = \nu_n$ for all $y \in \sY$. Since $\rL_q(\mu,\G)$ is relatively (sequentially) compact with respect to the Young narrow topology, there exists a subsequence $\{\gamma_{n_k}\}$ of $\{\gamma_n\}$ that converges with respect to the Young narrow topology. This implies that subsequence $\{\nu_{n_k}\}$ of $\{\nu_n\}$ also converges with respect to the weak topology. Hence, $\G$ is relatively  (sequentially) compact with respect to the weak topology. By Prohorov's theorem \cite[Theorem 1.4.12]{HeLa03}, $\G$ is tight. This completes the proof of item (iii).

The proofs of items (i) and (ii) are somewhat similar. Since the proof of item (i) is also rather standard, we present only the proof of item (ii). First, suppose that $\G$ is tight. Let $\conv L_q(\mu,\H_{\G})$ be the convex hull of $L_q(\mu,\H_{\G})$, which is a subset of $L_q\bigl(\mu,\H_{\P}\bigr)$ as $L_q\bigl(\mu,\H_{\P}\bigr)$ is convex. Since $q \in (1,\infty)$, $L_q\bigl(\mu,\H_{k}\bigr)$ is reflexive. Therefore, weak$^*$-topology (i.e., kernel mean embedding topology) coincides with the weak-topology. Moreover, for convex sets, closedness with respect to the weak-topology and norm-topology are equivalent. Therefore, it is sufficient to prove that $\conv L_q(\mu,\H_{\G})$ is closed in norm-topology.

Let the sequence $\{\gamma_{\lambda}\}$ in $\conv L_q(\mu,\H_{\G})$ be convergent in norm-topology. Since $L_q\bigl(\mu,\H_{\P_{_{\leq1}}}\bigr)$  is closed with respect to the weak kernel mean embedding topology (and so in norm-topology), the limit $\gamma$ of $\{\gamma_{\lambda}\}$ must be in $L_q\bigl(\mu,\H_{\P_{_{\leq1}}}\bigr)$. Note that convergence in mean implies almost sure convergence along a sub-sequence, and so, there exists a sub-sequence $\{\gamma_{\theta}\}$ of $\{\gamma_{\lambda}\}$ such that for $\mu$-almost everywhere, we have 
$$
\|\gamma_{\theta}(y) - \gamma(y) \|_{\M} \rightarrow 0.
$$
Fix any $y$ that satisfies above convergence result, where these $y$ values have probability $1$. Let $\{h_n\}$ be a dense subset of $\H_k$. Since $\|\cdot\| \leq M \, \|\cdot\|_{\H_k}$, this set $\{h_n\}$ is also dense in $\H_k$ with respect to the sup-norm. Since $\H_k$ is assumed to be dense in $C_0(\sU)$, we can also conclude that $\{h_n\}$ is dense in $C_0(\sU)$. Since $\{h_n\}$ is dense in $C_0(\sU)$ and 
$$
\int_{\sY} \int_{\sU} h_n(u) \, \gamma_{\theta}(y)(du) \, \mu(dy) \rightarrow \int_{\sY} \int_{\sU} h_n(u) \, \gamma(y)(du) \, \mu(dy), \,\, \forall n,
$$
$I^{-1} \circ \gamma_{\theta}(y)$ converges vaguely to $I^{-1} \circ \gamma(y)$. Note that $\{I^{-1} \circ \gamma_{\theta}(y)\} \subset \G$ and $\G$ is tight, and so, $\{I^{-1} \circ \gamma_{\theta}(y)\}$ has a subsequence which converges to some probability measure in weak convergence topology (and also in vague topology) by Prohorov's theorem \cite[Theorem 1.4.12]{HeLa03}. This implies that $I^{-1} \circ \gamma(y)$ is a probability measure. Hence, $\gamma \in L_q(\mu,\H_{\P})$. Therefore, the set $\conv L_q(\mu,\H_{\G})$ has the closure with respect to the norm-topology (or equivalently, with respect to the weak kernel mean embedding topology) in $L_q(\mu,\H_{\P}) \equiv \Gamma$. This implies that $L_q(\mu,\H_{\G})$ has the closure with respect to the weak kernel mean embedding topology in $L_q(\mu,\H_{\P}) \equiv \Gamma$. Conversely, suppose that the set $L_q(\mu,\H_{\G})$ has the closure with respect to the weak kernel mean embedding topology in $L_q(\mu,\H_{\P}) \equiv \Gamma$ and $\sU$ is complete. Let $\{\nu_n\}$ be a sequence in $\G$ and for each $n$, define the constant function $\gamma_n(y) = I \circ \nu_n$ for all $y \in \sY$. Since $\rL_q(\mu,\H_{\G})$ is relatively (sequentially) compact with respect to the weak kernel mean embedding topology, there exists a subsequence $\{\gamma_{n_k}\}$ of $\{\gamma_n\}$ that converges with respect to the weak kernel mean embedding topology. This implies that subsequence $\{\nu_{n_k}\}$ of $\{\nu_n\}$ also converges with respect to the weak MMD topology. According to \cite[Lemma 5]{SiBaScMa24}, we also have convergence in the weak convergence topology, as the limit measure of ${\nu_{n_k}}$ is a probability measure, given that the closure of $L_q(\mu,\H_{\G})$ lies within $L_q(\mu,\H_{\P})$. Hence, $\G$ is relatively  (sequentially) compact with respect to the weak convergence topology. By Prohorov's theorem \cite[Theorem 1.4.12]{HeLa03}, $\G$ is tight. This completes the proof of item (ii).
\end{proof}

\section{Strong Kernel Mean Embedding Topology for Stochastic Kernels}

Weak topologies are primarily introduced to prove the existence of optimal policies in decision-making problems under the most general and mild assumptions. Their mathematical structure is well-suited for demonstrating such existence results, as weak topologies simplify the handling of compactness and continuity properties in infinite-dimensional spaces. Consequently, it is both natural and effective to define weak topologies on the set of policies. For instance, the Young narrow topology was originally defined to establish the existence of optimal relaxed policies in continuous-time deterministic optimal control problems \cite{warga2014optimal}. Similarly, the \(w^*\)-topology on the set of stochastic kernels was developed in \cite[Section 2.4]{BoArGh12} to analyze randomized Markov policies and prove the existence of optimal policies for continuous-time stochastic control problems under the average cost optimality criterion. More recently, in \cite{Sal20}, the \(w^*\)-topology was employed to establish the existence of optimal policies in team decision problems (see also \cite{yuksel2023borkar,SaYu22}). We note that the weak kernel mean embedding topology is equivalent to both the \( w^* \)-topology and the Young narrow topology for proving the existence of optimal solutions, as established in Theorem~\ref{equivalence} and Theorem~\ref{tightness}. Consequently, all results demonstrated in the aforementioned papers also hold under the weak kernel mean embedding topology.

However, while weak topologies are sufficient for establishing the existence of optimal solutions, they are often inadequate for addressing more complex tasks, such as approximation, robustness, and learning in the presence of model uncertainties. These challenges require stronger analytical tools to handle perturbations and approximations effectively. In such cases, strong norm topologies become indispensable. Unlike weak topologies, strong norm topologies provide the precision necessary to address issues of approximation and robustness. As a result, it is both natural and practical to define strong norm topologies on the set of stochastic dynamics rather than on policies. For example, in Markov decision problems, it is more appropriate to use a weak topologies on the set of policies to establish the existence of an optimal policy. However, for tasks like learning, approximation, and robustness analysis, it is essential to apply a strong norm topologies to the set of transition probabilities that describe the system dynamics. These two kind of topologies serve different purposes.

Therefore, in this section, we introduce the strong kernel mean embedding topology. This topology offers a robust framework for addressing model uncertainties, extending our analysis beyond existence results to include practical requirements such as approximation, robustness, and learning. It complements weak topologies by bridging the gap between theoretical foundations and real-world applications.  

\begin{tcolorbox} 
[colback=white!100]
\begin{definition}[Strong Kernel Mean Embedding Topology]\label{StrongKMET}
Strong kernel mean embedding topology on $$L_q(\mu,\H_{\P}) = I \circ \Gamma \, (\equiv \Gamma)$$ is the relative strong norm topology of $L_q(\mu,\H_{k})$; that is, $\gamma_{\lambda} \rightarrow \gamma$ in $L_q(\mu,\H_{\P})$, if
\begin{align*}
\begin{cases}
\int_{\sY} \|\gamma_{\lambda}(y)-\gamma(y)\|_{\M}^q \, \mu(dy) \rightarrow 0, & \,\, \text{if} \,\, q < \infty \\
\ess \sup_{y \in \sY} \|\gamma_{\lambda}(y)-\gamma(y)\|_{\M} \rightarrow 0, & \,\, \text{if} \,\, q = \infty.
\end{cases}
\end{align*}
\end{definition}
\end{tcolorbox}

Conditions requiring the relative compactness of \( L_q(\mu, \H_{\P}) \) with respect to the strong kernel mean embedding topology is generally quite demanding \cite{DiMa99} typically requring generalizations of the Arzela-Ascoli theorem in a Bochner space formulation, which then imposes uniform regularity conditions. As a result, this topology is not well-suited for certain applications, such as on establishing the existence of optimal solutions in decision problems. However, a trade-off will occur in view of continuity at the expense of compactness: In the context of learning and robustness, where uncertainty in the system dynamics is a central concern, the strong kernel mean embedding topology proves to be the appropriate choice for analyzing the set of stochastic dynamics. Along this direction, in the next section, we examine the robustness problem in Markov decision processes. For this problem, even without imposing specific conditions on the perturbed models, we can establish robustness results: if the perturbed models converge to the true model under the strong kernel mean embedding topology, the optimal policies of the perturbed models remain nearly optimal for the original model. In contrast, achieving similar results under the weak kernel mean embedding topology requires additional assumptions on the stochastic dynamics of the perturbed models. These extra conditions significantly limit the applicability of the results to real-world problems, underscoring the practical importance of the strong kernel mean embedding topology in such settings.

\begin{remark}
The weak convergence topology on the set of probability measures can be metrized using various metrics, one of which is the bounded-Lipschitz metric, denoted by \(\rho_{BL}\) \cite[p. 394]{Dud89}. Using this metric, we can define the following strong topology, similar to the strong kernel mean embedding topology, on the set of stochastic kernels \(\Gamma\): \(\gamma_{\lambda} \rightarrow \gamma\) in \(\Gamma\) if  
\[
\begin{cases}  
\int_{\sY} \rho_{BL}(\gamma_{\lambda}(y),\gamma(y))^q \, \mu(dy) \rightarrow 0, & \text{if } q < \infty, \\  
\ess \sup_{y \in \sY} \rho_{BL}(\gamma_{\lambda}(y),\gamma(y)) \rightarrow 0, & \text{if } q = \infty.  
\end{cases}  
\]  
However, an important distinction exists between this topology and the strong kernel mean embedding topology. The latter is a relative topology, meaning it can also be defined on the set of functions from \(\sY\) to \(\M(\sU)\) (not necessarily only for stochastic kernels) since the MMD topology is well-defined for all finite signed measures. In contrast, the weak convergence topology is not metrizable on the set of all finite signed measures, making it impossible to extend the above definition to such functions. Therefore, in this case, the machinery of Bochner space theory, which is typically useful for establishing results on approximation, robustness, and learning using tools from functional analysis, cannot be applied. More specifically, the $L_p$-space theory of vector-valued functions, which is quite powerful, is not applicable here. 

To extend this notion to the set of functions from \(\sY\) to \(\M(\sU)\), one would need to use the total variation norm in place of the bounded-Lipschitz metric. However, this results in a topology that is significantly stronger than the strong kernel mean embedding topology and is often impractical, particularly when dealing with empirical estimates of stochastic dynamics. Consequently, the strong kernel mean embedding topology offers a distinct advantage over topologies defined using the bounded-Lipschitz metric or total variation norm, making it more suitable for practical applications.
\end{remark}

Below, we show that the weak kernel mean embedding topology does not imply the strong kernel mean embedding topology on the set of stochastic kernels. Clearly, the strong kernel mean embedding topology is stronger than the weak kernel mean embedding topology on the set of all $q$-Bochner integrable functions from $\sY$ to $\H_k$. However, if the range of the function is contained in $\H_{\P}$, the situation may differ. The following result demonstrates that this is not the case.

\begin{theorem}\label{strongStrongerthanWeakThm}
Let $\gamma_{\lambda} \rightharpoonup^* \gamma$ in $L_q(\mu,\H_{\P})$. Then, it is not necessarily the case that $\gamma_{\lambda} \rightarrow \gamma$ in strong kernel mean embedding topology.
\end{theorem}

\begin{proof}
It suffices to provide a counterexample. We build on an example from \cite{YukselOptimizationofChannels} (used in a different context) in the following. Let $\sY=[0,1]$, $\sU=\{0,1\}$, and $\mu$ be the Lebesgue measure (uniform distribution) on $[0,1]$. Let
\begin{eqnarray}\label{star2}
L_{nk}= \left[\frac{2k-2}{2n},\frac{2k-1}{2n}\right), \quad
    R_{nk}= \left[\frac{2k-1}{2n},\frac{2k}{2n}\right)
\end{eqnarray}
and define the {\it square wave} function
\[
  h_n(y) = \sum_{k=1}^n\bigl(1_{\{y \in L_{nk}\}} - 1_{\{y \in R_{nk}\}} \bigr).
\]
Define further $f_n(y)=h_n(y)+1$ and $\tilde f_n(y)=1-h_n(y)$. 

 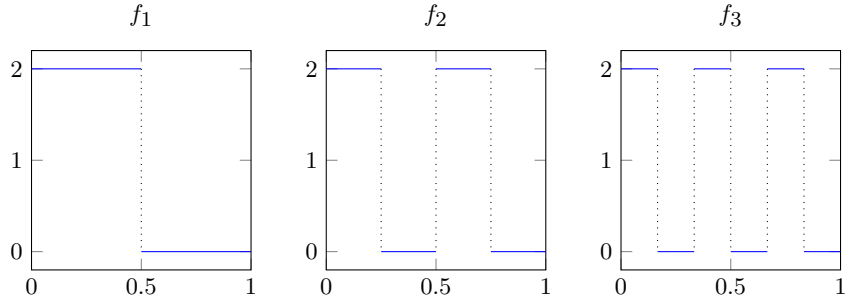
\begin{figure}[H]
      \centering
      \begin{tikzpicture}
        \begin{groupplot}[group style={group size=3 by 1},xmin=0,xmax=1,ymin=-0.2,ymax=2.2,height=4.5cm,width=4.5cm,ytick={0,1,2},xtick={0,0.5,1},no markers]
          \nextgroupplot[title={$f_1$}];
          \addplot[domain=0:0.5,blue] {2};
          \addplot[domain=0.5:1,blue] {0};
          \draw[dotted] (axis cs:0.5,2) -- (axis cs:0.5,0);
          \nextgroupplot[title={$f_2$}];
          \addplot[domain=0:0.25,blue] {2};
          \addplot[domain=0.25:0.5,blue] {0};
          \addplot[domain=0.5:0.75,blue] {2};
          \addplot[domain=0.75:1,blue] {0};
          \draw[dotted] (axis cs:0.25,2) -- (axis cs:0.25,0);
          \draw[dotted] (axis cs:0.5,2) -- (axis cs:0.5,0);
          \draw[dotted] (axis cs:0.75,2) -- (axis cs:0.75,0);
          \nextgroupplot[title={$f_3$}];
          \addplot[domain=0:1/6,blue] {2};
          \addplot[domain=1/6:2/6,blue] {0};
          \addplot[domain=2/6:3/6,blue] {2};
          \addplot[domain=3/6:4/6,blue] {0};
          \addplot[domain=4/6:5/6,blue] {2};
          \addplot[domain=5/6:1,blue] {0};
          \draw[dotted] (axis cs:1/6,2) -- (axis cs:1/6,0);
          \draw[dotted] (axis cs:2/6,2) -- (axis cs:2/6,0);
          \draw[dotted] (axis cs:3/6,2) -- (axis cs:3/6,0);
          \draw[dotted] (axis cs:4/6,2) -- (axis cs:4/6,0);
          \draw[dotted] (axis cs:5/6,2) -- (axis cs:5/6,0);
        \end{groupplot}
      \end{tikzpicture}
      \caption{Functions $f_1$, $f_2$, and $f_3$.}\label{fig:squarewave}
    \end{figure}
    
By the Riemann-Lebesgue lemma (\cite{WhZy77}, Thm.\ 12.21), for any $g: \sY \rightarrow \R$, we have
$$
\int_{\sY} g(y) \, h_n(y) \, \mu(dy) \rightarrow 0. 
$$
Hence 
$$
\frac{1}{2} \, \int_{\sY} g(y) \, f_n(y) \, \mu(dy), \, \frac{1}{2} \, \int_{\sY} g(y) \, \tilde f_n(y) \, \mu(dy) \rightarrow \frac{1}{2} \, \int_{\sY} g(y) \, \mu(dy). 
$$    
Let $\gamma_n(1|y) \triangleq 1_{\{y \in L_n^k\}}$ and $\gamma_n(0|y) \triangleq 1_{\{y \in R_n^k\}}$. For any $h \in L_p(\mu,\H_k)$, by above discussion,  we have 
\begin{align*}
&\int_{\sY} \int_{\sU} h(y)(u) \, \gamma_n(y)(du) \, \mu(dy) = \int_{\sY} h(y)(1) \, f_n(y) \, \mu(dy) + \int_{\sY} h(y)(0) \, \tilde f_n(y) \, \mu(dy) \\
&\phantom{xxxxxxxxxxxxx}\rightarrow \int_{\sY} \frac{1}{2} \{h(y)(1)+h(y)(0)\} \, \mu(dy) \triangleq \int_{\sY} \int_{\sU} h(y)(u) \, \gamma(y)(du) \, \mu(dy),
\end{align*}
where $\gamma(y) \triangleq 1/2 \, \delta_0 + 1/2 \, \delta_1$. Hence, $\gamma_{n} \rightharpoonup^* \gamma$ in $L_q(\mu,\H_{\P})$. However, for all $n$ and for any $y \in \sY$, we have 
\begin{align*}
\|\gamma_n(y) - \gamma(y)\|_{\M} = \frac{1}{2} \, \|\delta_0-\delta_1\|_{\M} > 0. 
\end{align*}
Hence $\gamma_n$ does not converge to $\gamma$ in strong kernel mean embedding topology.
\end{proof}

As demonstrated in the previous theorem, convergence in the weak kernel mean embedding topology does not imply convergence in the strong kernel mean embedding topology for stochastic kernels. However, under the assumption of bounded and equicontinuous densities for the probability measures in the stochastic kernels, it is possible to establish this implication from the weak to the strong version. Notably, similar conditions are also used to ensure the robustness of Markov decision processes under the weak kernel mean embedding topology in the next section.

\begin{theorem}\label{WeakImpliesStrong}
Let $\gamma_{\lambda} \rightharpoonup^* \gamma$ in $L_q(\mu,\H_{\P})$. Moreover, there exists a fixed probability measure $\xi_{\abss}$ on $\sU$ such that $\gamma_{\lambda}(y)(\cdot), \gamma(y)(\cdot) \ll \xi_{\abss}$, for all $y \in \sY$ and for all $\lambda$ with the following density functions $f_{\lambda}(u,y)$ and $f(u,y)$. We suppose that 
\begin{itemize}
\item [(a)] The set of densities $\{f_{\lambda}(u,y)\}$ is bounded, i.e.,
$$
\sup_{\lambda} \sup_{(u,y) \in \sU\times\sY} |f_{\lambda}(u,y)| < \infty.
$$
\item [(b)] The set $\{f_{\lambda}(\cdot,\cdot)\}_{\lambda}$ is equicontinuous with respect to $(u,y)$. 
\end{itemize}
Then, $\gamma_{\lambda} \rightarrow \gamma$ in strong kernel mean embedding topology for $q < \infty$.
\end{theorem}

\begin{proof}
Suppose that the result is not true. Then, there exists a subsequence $\{\gamma_{\tau}\}$ of $\{\gamma_{\lambda}\}$ such that the kernel mean embedding strong norm of the difference $\gamma_{\tau} - \gamma$ is \sy{strictly greater than some $\epsilon > 0$} for all $\tau$. Under the stated assumptions and by Arzela-Ascoli theorem \cite[Theorem 4.44]{Roy06}, there exists a subsequence $\{f_{\theta}\}$ of $\{f_{\tau}\}$ such that $f_{\theta} \rightarrow \tilde f$ uniformly over compact subsets, for some continuous and bounded density function $\tilde f$. It then follows that $\gamma_{\theta} \rightarrow \tilde \gamma$ with respect to MMD topology uniformly over compact subsets, where $\tilde \gamma(y)(du) \triangleq \tilde f(u,y) \, \xi_{\abss}(du)$. Since $\gamma_n \rightharpoonup^* \gamma$, we can also conclude that $\gamma = \tilde \gamma$ $\mu$-almost-surely. As uniform convergence over compact sets is stronger than pointwise convergence, we have
$$
\|\gamma_{\theta}(y)(\cdot) - \gamma(y)(\cdot)\|_{\M} \rightarrow 0 \,\, \text{$\mu$-almost-surely}. 
$$
This implies the convergence in strong kernel mean embedding topology for $q < \infty$ of $\gamma_{\theta}$ to $\gamma$ as the MMD norm of the difference of two probability measures is bounded via Assumption~\ref{as2}-(a). This contradicts with the initial assumption, which completes the proof. 
\end{proof}

\begin{remark}
Although condition (b) in the theorem is not overly restrictive, as explained in Remark~\ref{equicont-assumption}, it can be slightly relaxed using the Arzelà-Ascoli theorem for vector-valued functions \cite[Theorem 1.6.14]{PaWi18}. Specifically, we aim to show that the mappings \(\gamma_{\lambda}: \sY \rightarrow \H_{\P}\) are relatively sequentially compact under the topology of uniform convergence on compacta. To achieve this, two requirements must be met: 
\begin{enumerate}
    \item[(i)] The set \(\{\gamma_{\lambda}: \sY \rightarrow \H_{\P}\}\) should be equicontinuous.
    \item[(ii)] For any \(y \in \sY\), the set \(\{\gamma_{\lambda}(y)(\cdot)\} \subset \H_{\P}\) must be relatively compact.
\end{enumerate}
Since the MMD topology is equivalent to the weak convergence topology on \(\H_{\P}\) \cite[Lemma 5 and Remark 6]{SiBaScMa24}, the latter condition is equivalent to the tightness of \(\{\gamma_{\lambda}(y)(\cdot)\}\) by Prokhorov's theorem \cite[Theorem 1.4.12]{HeLa03}. Consequently, condition (b) in the theorem can be replaced with the following slightly relaxed version:
\begin{itemize}
    \item [(b')] The set \(\{f_{\lambda}(u,\cdot)\}_{\lambda, u}\) is equicontinuous with respect to \(y\). Moreover, for any \(y \in \sY\), the set \(\{f_{\lambda}(\cdot,y)\}_{\lambda}\) is equicontinuous with respect to \(u\).
\end{itemize}
The proof of this relaxation is straightforward, albeit more detailed, so we omit it here. The first condition in (b') directly implies (i), while the second condition ensures (ii). It is clear, however, that conditions (b) and (b') are not significantly different in practice.
\end{remark}

{\color{black}
\begin{remark}
The result in Theorem~\ref{WeakImpliesStrong} can alternatively be established under a more direct assumption:
\[
f_{\lambda}(\cdot,y) \to f(\cdot,y) \quad \xi_{\abss}\text{-a.s.}
\]
for all \( y \in \sY \), except on a set of $\mu$-measure zero. In that case, Scheffe's theorem~\cite[Theorem 16.12]{Bil95} can be directly applied to obtain the desired conclusion. It is worth emphasizing that, under the assumptions (a) and (b) of Theorem~\ref{WeakImpliesStrong}, the Arzela-Ascoli theorem guarantees uniform convergence of densities over compact subsets along a subsequence. However, such uniform convergence on compact sets does not, in general, imply almost sure pointwise convergence along the original sequence -- simple counterexamples can be constructed. In Theorem~\ref{WeakImpliesStrong}, we additionally assume that the stochastic kernels converge in the weak kernel mean embedding topology, which ensures that the corresponding sequence of densities admits a unique cluster point almost surely. Combined with the Arzela-Ascoli theorem, this implies almost sure pointwise convergence, which is precisely the condition required by Scheffe's theorem. Hence, the approach based directly on Scheffe's theorem yields a stronger result than ours. Nevertheless, the assumptions in Theorem~\ref{WeakImpliesStrong} are easier to verify in practice, which motivates stating the weaker, but more accessible, version presented in this paper.
\end{remark}
}

\section{Applications}\label{applications}

\subsection{Robustness of Markov Decision Processes under Strong Kernel Topologies}

In this section, we provide an application on compactness, approximability and robustness of policies obtained for an approximate MDP model under the strong and weak formulations.

A discrete-time Markov decision process (MDP) can be described by a five-tuple
\begin{align}
\bigl( \sX, \sA, p, c, \kappa_0 \bigr), \nonumber
\end{align}
where Borel space (i.e., Borel subset of complete and separable metric space) $\sX$ denotes the \emph{state} space and compact Borel space $\sA$ denotes the \emph{action} space.  The \emph{stochastic kernel} $p(\,\cdot\,|x,a)$ denotes the \emph{transition probability} of the next state given that previous state-action pair is $(x,a)$. The \emph{one-stage cost} function $c$ is a measurable function from $\sX \times \sA$ to $\R_{+}$. The \emph{initial state distribution} is $\kappa_0$.

Define the history spaces $\sH_0 = \sX$ and
$\sH_{t}=(\sX\times\sA)^{t}\times\sX$, $t=1,2,\ldots$ endowed with their
product Borel $\sigma$-algebras generated by $\B(\sX)$ and $\B(\sA)$. A
\emph{policy} is a sequence $\pi=\{\pi_{t}\}$ of stochastic kernels
on $\sA$ given $\sH_{t}$. The set of all policies is denoted by $\Pi$.
Let $\Phi$ denote the set of stochastic kernels $\varphi$ on $\sA$ given $\sX$, and let $\rF$ denote the set of all measurable functions $f$ from $\sX$ to $\sA$. A \emph{randomized Markov} policy is a sequence $\pi=\{\pi_{t}\}$ of stochastic kernels on $\sA$ given $\sX$. A \emph{deterministic Markov} policy is a sequence of stochastic kernels $\pi=\{\pi_{t}\}$ on $\sA$ given $\sX$ such that $\pi_{t}(\,\cdot\,|x)=\delta_{f_t(x)}(\,\cdot\,)$ for some $f_t \in \mathbb{F}$, where $\delta_z$ denotes the point mass at $z$. The set of randomized and deterministic Markov policies are denoted by $\sR\sM$ and $\sM$, respectively. A \emph{randomized stationary} policy is a
constant sequence $\pi=\{\pi_{t}\}$ of stochastic kernels on $\sA$ given $\sX$ such that
$\pi_{t}(\,\cdot\,|x)=\varphi(\,\cdot\,|x)$ for all $t$ for some
$\varphi \in \Phi$. A \emph{deterministic stationary} policy is a constant sequence of stochastic kernels $\pi=\{\pi_{t}\}$ on $\sA$ given $\sX$ such that $\pi_{t}(\,\cdot\,|x)=\delta_{f(x)}(\,\cdot\,)$ for all $t$ for some
$f \in \mathbb{F}$. The set of randomized and deterministic stationary policies are identified with the sets $\Phi$ and $\mathbb{F}$, respectively.

According to the Ionescu Tulcea theorem (see \cite{HeLa96}), an initial distribution $\kappa_0$ on $\sX$ and a policy $\pi$ define a unique probability measure $P_{\mu}^{\pi}$ on $\sH_{\infty}=(\sX\times\sA)^{\infty}$.
The expectation with respect to $P_{\kappa_0,p}^{\pi}$ is denoted by $\cE_{\kappa_0,p}^{\pi}$. The cost functions to be minimized in this paper are the $\beta$-discounted cost given by
\begin{align}
J(\pi;\kappa_0,p) &= \cE_{\kappa_0,p}^{\pi}\biggl[\sum_{t=0}^{\infty}\beta^{t}c(x_{t},a_{t})\biggr]. \nonumber
\end{align}
With this notation, the discounted optimal value function of the control problem is defined as
\begin{align}
J^*(\kappa_0,p) &\coloneqq \inf_{\pi \in \Pi} J(\pi,x). \nonumber
\end{align}
A policy $\pi^{*}$ is said to be optimal if $J(\pi^{*};\kappa_0,p) = J^*(\kappa_0,p)$. Under fairly mild conditions, the set $\rF$ of deterministic stationary policies contains an optimal policy
for discounted cost (see, e.g., \cite{HeLa96,FeKaZa12}). Hence, in the remainder of this section, we only focus on deterministic stationary policies. 

We impose the assumptions below on the components of the Markov decision process.

\begin{assumption}
\label{MDP:as1}
\begin{itemize}
\item [  ]
\item [(a)] The one-stage cost function $c$ is in $C_b(\sX \times \sA)$.
\item [(b)] $\sX$ is locally compact metric space and $\sA$ is finite.
\end{itemize}
\end{assumption}

Note that if $\sX = \R^n$ for some $n \geq 1$, then $\sX$ is locally compact metric space under Euclidean or one of the equivalent metrics. Moreover, by \cite[Theorem 3.2]{SaYuLi16JMAA}, under the weak continuity of the transition kernel $p$ (i.e., $p(\cdot|x_n,a_n) \rightarrow p(\cdot|x,a)$ in weak convergence topology for any $(x_n,a_n) \rightarrow (x,a)$ in $\sX\times\sA$), any MDP with compact action space can be approximated by MDPs with finite action spaces. Therefore, finite action space assumption is not restrictive in terms of optimality.

Let $\Gamma$ denote the set of stochastic kernels from $\sX\times\sA$ to $\sX$. Let $\kappa_{\abss} \in \P(\sX)$. Define 
$$
\sM_{\abss} \triangleq \{p \in \Gamma: p(\cdot|x,a) \ll \kappa_{\abss} \,\, \forall (x,a) \in \sX\times\sA \} \subset \Gamma. 
$$
Hence, if $p \in \sM_{\abss}$, then there exists a measurable function $f:\sX\times\sX\times\sA \rightarrow \R_{+}$ such that for any $(x,a) \in \sX\times\sA$, we have $p(dy|x,a) = f(y|x,a) \, \kappa_{\abss}(dy)$, where $f(\cdot|x,a)$ is Radon-Nikodym derivative of $p(\cdot|x,a)$ with respect to $\kappa_{\abss}$. 

In the robustness problem, we have an unknown true transition kernel \( p \) and a sequence of approximations \(\{p_n\}\), which are constructed from available data. We know that \( p_n \) converges to \( p \) as \( n \to \infty \) in a certain sense. Our goal is to establish the following results:
\begin{itemize}
\item[(P1)] $\lim_{n\rightarrow\infty} J^*(\kappa_0,p_n) = J^*(\kappa_0,p)$; that is, the optimal value of the approximate model converges to the optimal value of the true model.
\item[(P2)] For each $n$, let $\pi_n^*$ be the optimal deterministic stationary policy of the approximate model. Then, $\lim_{n\rightarrow\infty} J(\pi_n^*;\kappa_0,p) = J^*(\kappa_0,p)$; that is, the optimal policy of the approximate model is nearly optimal for the true model.
\end{itemize} 
To establish above results, we put the following conditions on transition kernels. 

\begin{assumption}
\label{MDP:as2}
\begin{itemize}
\item [  ]
\item [(a)] $\{p_n\}$ and $p$ are weakly continuous.
\item [(b)] $p \in \sM_{\abss}$ and $p_n \in \sM_{\abss}$ for all $n$. 
\item [(c)] $\kappa_0 \ll \kappa_{\abss}$.
\end{itemize}
\end{assumption}

Under Assumption~\ref{MDP:as2}-(b), for all $(x,a)$, $p(\cdot|x,a)$ and $p_n(\cdot|x,a)$ have densities with respect to $\kappa_{\abss}$. Let us denote these densities by $f(y|x,a)$ and $f_n(y|x,a)$, respectively. To establish the robustness results, we first use the strong kernel mean embedding topology. Then, by putting further conditions on above densities, we also establish the same results under weak kernel mean embedding topology. Indeed, weak topologies are introduced and used to establish the existence of optimal policies in stochastic control. Therefore, they are in general quite weak for establishing robustness results. That is why we need strong kernel mean embedding topology. However, under some further regularity conditions on the densities, we can still establish the robustness result under weak kernel mean embedding topology.

\subsubsection{Robustness Under Strong Kernel Mean Embedding Topology}\label{secStrongKernelRobustness}

In this section, we suppose that Assumptions~\ref{MDP:as1} and \ref{MDP:as2} hold. Define $\mu_{\reff}(dx,a) \triangleq \kappa_{\abss}(dx) \otimes \sUnif(a)$, where $\sUnif$ is the uniform distribution on $\sA$. Let 
$$k:\sX \times \sX \rightarrow \R$$
be a positive semi-definite kernel that induces the reproducing kernel Hilbert space $\H_k$. We suppose that Assumption~\ref{as2} is true for RKHS $\H_k$ and its kernel $k$. Therefore, we can endow the set of transition kernels $\Gamma$ with strong kernel mean embedding topology with $p=\infty$; that is, $\Gamma \subset L_{\infty}(\mu_{\reff},\H_k)$. To establish the robustness results, we suppose that the following convergence condition holds:

\begin{align}
\|p_n-p\|_{L_{\infty}(\mu_{\reff},\H_k)} \triangleq \esssup_{(x,a) \in \sX \times \sA} \|p_n(\cdot|x,a) - p(\cdot|x,a)\|_{\M} \rightarrow 0, \label{conv1}
\end{align}
where the essential supremum is taken with respect to $\mu_{\reff}$. Hence, in view of this convergence, one can prove that there exists a set $C \subset \sX\times\sA$ such that $\mu_{\reff}(C)=1$ and 
$$
\sup_{(x,a) \in C} \|p_n(\cdot|x,a) - p(\cdot|x,a)\|_{\M} \rightarrow 0. 
$$
Since $\sA$ is finite, one can take the set $C$ in the following form $D \times \sA$, where $\kappa_{\abss}(D) =1$. Hence, 
$$
\sup_{(x,a) \in D\times\sA} \|p_n(\cdot|x,a) - p(\cdot|x,a)\|_{\M} \rightarrow 0. 
$$
Note that for any $(x,a) \in \sX\times\sA$, since $p_n(\cdot|x,a) \ll \kappa_{\abss}$ and $\kappa_{\abss}(D)=1$, we have 
$$
p_n(B|x,a) = \int_{B^c} f_n(y|x,a) \, \kappa_{\abss}(dy) = 0,
$$ 
when $B \subset D^c$. Hence, under $p_n$, the states of the control model cannot escape to $D^c$. The same is true for $p$. Additionally, since $\kappa_0 \ll \kappa_{\abss}$, with probability $1$, the initial state cannot lie in $D^c$. These two facts imply that our control model (both with $p_n$ and $p$) is equivalent to the control model with state space $D$. Therefore, without loss of generality, we can suppose that 
$$
\sup_{(x,a) \in \sX\times\sA} \|p_n(\cdot|x,a) - p(\cdot|x,a)\|_{\M} \rightarrow 0;
$$
that is, $p_n$ converges to $p$ uniformly with respect to MMD topology. Since MMD topology is equivalent to weak convergence topology under Assumption~\ref{as2} \cite[Lemma 5 and Remark 6]{SiBaScMa24}, we have uniform convergence with respect to the weak convergence topology. But this uniform convergence implies the following: 
$$
p_n(\cdot|x_n,a_n) \rightarrow p(\cdot|x,a) \,\, \text{weakly whenever} \,\, (x_n,a_n) \rightarrow (x,a).  
$$
The final condition, known as continuous convergence, is equivalent to uniform convergence over compact sets when the kernels are weakly continuous. Thus, condition (\ref{conv1}) is slightly stronger than necessary for establishing robustness. We now summarize the implications of the above discussion within the context of the robustness problem.

\begin{theorem}{\cite[Theorem 4.2, Theorem 4.4]{KaYu20}}\label{Ali_thm}
Suppose that $\{p_n\}$ and $p$ are weakly continuous and $p_n(\cdot|x_n,a_n) \rightarrow p(\cdot|x,a) \,\, \text{weakly whenever} \,\, (x_n,a_n) \rightarrow (x,a)$. Then, we have
\begin{itemize}
\item[(a)] $\lim_{n\rightarrow\infty} J^*(\kappa_0,p_n) = J^*(\kappa_0,p)$. 
\item[(b)] $\lim_{n\rightarrow\infty} J(\pi_n^*;\kappa_0,p) = J^*(\kappa_0,p)$ if $\pi_n^*$ is optimal under $p_n$ for any initial distribution. 
\end{itemize}
\end{theorem}

The condition on \(\pi_n^*\) in (b) above is generally non-restrictive and typically holds, as for discounted MDPs under any transition probability, the optimal deterministic stationary policy is derived via the dynamic programming principle. By construction, these policies are optimal regardless of the initial distribution, so condition (b) is satisfied in practice.

In the theorem above, the most critical condition is the continuous convergence of \( p_n \) to \( p \) in the weak convergence topology. This can be ensured through convergence in the strong kernel mean embedding topology when \( p = \infty \). 

In the next section, by introducing regularity conditions on the densities \(\{f_n\}\) and \(f\), we establish the continuous convergence of the transition kernels under the weak kernel mean embedding topology.

\subsubsection{Robustness Under Weak Kernel Mean Embedding Topology}\label{robust_kernel}

We suppose that the following convergence condition holds:
\begin{align}
p_n \rightharpoonup^* p \label{conv2}
\end{align}
 with respect to the weak kernel mean embedding topology for some $p \geq 1$ (not necessarily $p=\infty$). In addition to previous assumptions, we impose the following conditions on the densities. 
 
 \begin{assumption}
\label{MDP:as3}
\begin{itemize}
\item [  ]
\item [(a)] The set of densities $\{f_n(y|x,a)\}$ is bounded, i.e.,
$$
\sup_{n\geq1} \sup_{(y,x,a) \in \sX\times\sX\times\sA} |f_n(y|x,a)| < \infty.
$$
\item [(b)] The set $\{f_n(\cdot,\cdot,\cdot)\}_n$ is equicontinuous with respect to $(y,x,a)$ (this is true if densities are elements of bounded subset of some reproducing kernel Hilbert space). 
\end{itemize}
\end{assumption}

\begin{remark}\label{equicont-assumption}
Assumption~\ref{MDP:as3} may appear restrictive; however, if the approximate densities belong to a bounded subset of a reproducing kernel Hilbert space defined on \(\sZ \triangleq \sX \times \sX \times \sA\), then, under the same assumptions imposed on the \((k, \H_k)\) pair, Assumption~\ref{MDP:as3} holds for these densities. Indeed, let 
$$l:\sZ \times \sZ \rightarrow \R$$
be a positive semi-definite kernel that induces the reproducing kernel Hilbert space $\H_l$. We suppose that the conditions in Assumption~\ref{as2} are true for RKHS $\H_l$ and its kernel $l$; that is, $l$ is bounded and continuous, $l(\cdot,z) \in C_0(\sZ)$ for all $z \in \sZ$, and $\H_l$ is dense in $C_0(\sZ)$. For some $L > 0$, define 
$$
B \triangleq \{f \in \H_l: \|f\|_{\H_k} \leq L\}. 
$$
If $\{f_n\} \subset B$, then $\{f_n\}$ is equicontinuous and uniformly bounded. To see this, fix any $(z,y) \in \sZ$. Then we have 
\begin{align*}
|f(z)-f(y)| &= |\langle f,l(\cdot,z) \rangle_{\H_l} - \langle f,l(\cdot,y) \rangle_{\H_l}| \\
&= |\langle f,l(\cdot,z)-l(\cdot,y) \rangle_{\H_l}| \\
&\leq \|f\|_{\H_l} \, \|l(\cdot,z)-l(\cdot,y)\|_{\H_l}. 
\end{align*}
Note that 
\begin{align*}
\|l(\cdot,z)-l(\cdot,y)\|_{\H_l}^2 &= \langle l(\cdot,z)-l(\cdot,y),l(\cdot,z)-l(\cdot,y)\rangle_{\H_l} \\
&= \left( l(z,z) - l(z,y) \right) + \left( l(y,y) - l(y,z) \right). 
\end{align*}
Since $l$ is bounded and continuous, for any $\varepsilon > 0$, there exists $\delta > 0$ such that when $d_{\sZ\times\sZ}\left((z_1,z_2),(y_1,y_2)\right) \leq \delta$, we have $|l(z_1,z_2)-l(y_1,y_2)| \leq \varepsilon$. Then, using above inequality, we conclude that for any $f \in B$, when $d_{\sZ}(z,y) \leq \delta$, we have $|f(z)-f(y)| \leq L \, \sqrt{2 \varepsilon}$. Hence, $B$ (and so $\{f_n\}$) is equicontinuous. Uniform boundedness is obvious.

Above discussion also clarifies why the MMD topology is equivalent to the weak convergence topology on the set of probability measures. If \(\eta_n \to \eta\) in the weak convergence topology, this implies uniform convergence of integrals over any equicontinuous and uniformly bounded family of functions. Since the set \(\{f \in \H_l : \|f\|_{\H_l} \leq 1\}\) is equicontinuous and uniformly bounded, it follows that \(\eta_n\) converges to \(\eta\) in the MMD topology. The reverse implication is straightforward because \(\H_l\) is dense in \(C_0(\sZ)\). This completes the proof.
\end{remark}

We now state the main result of this subsection. 

\begin{theorem}\label{robust-weak}
Suppose that Assumptions~\ref{MDP:as1}, \ref{MDP:as2}, and \ref{MDP:as3} hold. If $p_n \rightharpoonup^* p$, then we have
\begin{itemize}
\item[(a)] $\lim_{n\rightarrow\infty} J^*(\kappa_0,p_n) = J^*(\kappa_0,p)$. 
\item[(b)] $\lim_{n\rightarrow\infty} J(\pi_n^*;\kappa_0,p) = J^*(\kappa_0,p)$ if $\pi_n^*$ is optimal under $p_n$ for any initial distribution. 
\end{itemize}
\end{theorem}

\begin{proof}

We only prove part (a) as the proof of part (b) is similar. Suppose that part (a) is not true. Then, there exists a subsequence $\{p_{\lambda}\}$ of $\{p_n\}$ such that
$$
|J^*(\kappa_0,p_{\lambda}) - J^*(\kappa_0,p) | > 0 \,\, \forall \lambda. 
$$
Under Assumption~\ref{MDP:as3} and by Arzela-Ascoli theorem \cite[Theorem 4.44]{Roy06}, there exists a subsequence $\{f_{\theta}\}$ of $\{f_{\lambda}\}$ such that $f_{\theta} \rightarrow \tilde f$ uniformly over compact subsets, for some continuous and bounded density function $\tilde f$. It is then very simple to establish that $p_{\theta} \rightarrow \tilde p$ with respect to MMD topology uniformly over compact subsets, where $\tilde p(dy|x,a) \triangleq \tilde f(y|x,a) \, \kappa_{\abss}(dy)$. Since $p_n \rightharpoonup^* p$, we can also conclude that $p = \tilde p$ $\mu_{\reff}$-almost-surely. As uniform convergence over compact sets is stronger than pointwise convergence, we have
$$
\|p_{\theta}(\cdot|x,a) - p(\cdot|x,a)\|_{\M} \rightarrow 0 \,\, \text{$\mu_{\reff}$-almost-surely}. 
$$
Let us denote the set where this convergence holds by $D\times\sA$, where $\kappa_{\abss}(D) = 1$. Since our control model is equivalent to the control model with state space $D$ as $\kappa_{\abss}(D) = 1$, without loss of generality, we can conclude that 
$$
\|p_{\theta}(\cdot|x,a) - p(\cdot|x,a)\|_{\M} \rightarrow 0 \,\, \forall (x,a) \in \sX\times\sA.
$$
We now prove that $p_{\theta}$ also converges to $p$ continuously in MMD topology (or equivalently in weak convergence topology). To this end, we fix any $g \in C_b(\sX)$ and let $(x_{\theta},a_{\theta}) \rightarrow (x,a)$. Then, our goal is to prove that 
$$
\left|\int_{\sX} g(y) \, p_{\theta}(dy|x_{\theta},a_{\theta}) - \int_{\sX} g(y) \, p(dy|x,a) \right| \rightarrow 0. 
$$ 
We can bound above expression as follows:
\small
\begin{align*}
&\left|\int_{\sX} g(y) \, p_{\theta}(dy|x_{\theta},a_{\theta}) - \int_{\sX} g(y) \, p(dy|x,a) \right| \\
&\leq \left|\int_{\sX} g(y) \, p_{\theta}(dy|x_{\theta},a_{\theta}) - \int_{\sX} g(y) \, p_{\theta}(dy|x,a) \right| + \left|\int_{\sX} g(y) \, p_{\theta}(dy|x,a) - \int_{\sX} g(y) \, p(dy|x,a) \right|.
\end{align*}
\normalsize
The second term converges to zero as $p_{\theta}(\cdot|x,a) \rightarrow p(\cdot|x,a)$. To prove the convergence of the first term, by dominated convergence theorem, it is sufficient to establish that for any $y$, $|f_{\theta}(y|x_{\theta},a_{\theta})- f_{\theta}(y|x,a)| \rightarrow 0$. Fix any $y$. Suppose that $|f_{\theta}(y|x_{\theta},a_{\theta})- f_{\theta}(y|x,a)| \nrightarrow 0$. Then, there exists a subsequence such that 
$$
|f_{\xi}(y|x_{\xi},a_{\xi})- f_{\xi}(y|x,a)| > 0 \,\, \forall \xi.
$$
By Assumption~\ref{MDP:as3} and by Arzela-Ascoli theorem \cite[Theorem 4.44]{Roy06}, there exists a subsequence $\{f_{\alpha}(y|\cdot,\cdot)\}$ of $\{f_{\xi}(y|\cdot,\cdot)\}$ such that $f_{\alpha}(y|\cdot,\cdot) \rightarrow \tilde f(y|\cdot,\cdot)$ uniformly over compact subsets, for some continuous and bounded density function $\tilde f(y|\cdot,\cdot)$. Define the compact set $K \triangleq \{(x_{\alpha},a_{\alpha})\} \bigcup \{(x,a)\}$. Then we have 
\begin{align*}
&|f_{\alpha}(y|x_{\alpha},a_{\alpha})- f_{\alpha}(y|x,a)| \leq |f_{\alpha}(y|x_{\alpha},a_{\alpha})- \tilde f(y|x_{\alpha},a_{\alpha})| + |\tilde f(y|x_{\alpha},a_{\alpha})- f_{\alpha}(y|x,a)| \\
&\leq \sup_{(x,a) \in K}|f_{\alpha}(y|x,a)- \tilde f(y|x,a)| + |\tilde f(y|x_{\alpha},a_{\alpha})- f_{\alpha}(y|x,a)|.
\end{align*}
The first term converges to zero as $f_{\alpha}(y|\cdot,\cdot) \rightarrow \tilde f(y|\cdot,\cdot)$ uniformly over compact subsets. The second term converges to zero since $\tilde f(y|x_{\alpha},a_{\alpha}) \rightarrow \tilde f(y|x,a)$ and $f_{\alpha}(y|x,a) \rightarrow \tilde f(y|x,a)$. Hence, we have 
$$
|f_{\alpha}(y|x_{\alpha},a_{\alpha})- f_{\alpha}(y|x,a)| \rightarrow 0. 
$$
This is a contradiction. Hence, for any $y$, $|f_{\theta}(y|x_{\theta},a_{\theta})- f_{\theta}(y|x,a)| \rightarrow 0$. This implies that 
$$
\left|\int_{\sX} g(y) \, p_{\theta}(dy|x_{\theta},a_{\theta}) - \int_{\sX} g(y) \, p(dy|x,a) \right| \rightarrow 0. 
$$ 
Hence, $p_{\theta}$ also converges to $p$ continuously in MMD topology (or equivalently in weak convergence topology). Therefore, by part (a) of Theorem~\ref{Ali_thm}, we have
$$
|J^*(\kappa_0,p_{\theta}) - J^*(\kappa_0,p) | \rightarrow 0. 
$$
This contradicts with 
$$
|J^*(\kappa_0,p_{\lambda}) - J^*(\kappa_0,p) | > 0 \,\, \forall \lambda. 
$$
This completes the proof.
\end{proof}

\begin{remark}
Theorem~\ref{robust-weak} relies on continuous convergence of densities; that is,
\begin{align}
f_n(\cdot|x_n,a_n) \rightarrow f(\cdot|x,a),\label{cont-converge}
\end{align}
where $(x_n,a_n) \rightarrow (x,a)$ (see \cite[Proposition 2.8]{OmarModelLearningAverage} for a similar result under weak convergence of transition kernels). This is indeed equivalent to the uniform convergence of densities over compact sets when densities are continuous  and this is established through uniform boundedness, equicontinuity, and the Arzela-Ascoli theorem along a subsequence in the above proof. However, we note that the same result can alternatively be obtained by directly assuming, instead of Assumption~\ref{MDP:as3}, the continuous convergence of densities (\ref{cont-converge}) and applying Scheffe's theorem \cite[Theorem 16.12]{Bil95}.  
\end{remark}

\subsubsection{Implications for Asymptotic Optimality under Empirical Model Learning}

The implications of the results given above is that one can learn models via empirical data. In particular, there exist several studies which ensures that kernels can be learned from data, see e.g. \cite[Theorem 4]{GyorfiKohler07} for one of the earlier results under a stronger convergence criterion in the presence of empirical data; and \cite{TaBa24} and the references therein for an empirical convergence analysis under the strong kernel mean embedding topology. 

For an MDP with Borel state and action spaces, under weak continuity conditions, it is known that the action space can be approximated by a finite action set with arbitrarily small change in performance; see \cite{SaYuLi16JMAA} \cite[Theorem 3.16]{SaLiYuSpringer} for the discounted cost and \cite{SaYuLi16JMAA},\cite[Theorem 3.22]{SaLiYuSpringer} for average cost criterion. Thus, one needs to learn only finitely many kernels from data. Given the analysis presented in the previous section, these then lead to robustness to empirical learning in MDPs; see \cite{OmarModelLearningAverage,zhou2024robustness} (building on  \cite{KaYu20,kara2022robustness}) for an application.

\sy{
\begin{remark} [On Adapted Topologies] A topology on spaces of probability measures corresponding to laws of stochastic processes, which has been used in a wide variety of contexts in stochastic analysis, is defined by the following. A sequence of stochastic processes is said to converge to another process if their finite-dimensional marginals converge weakly, and their conditional distributions of future variables given the past (viewed as measure-valued stochastic processes) also converge weakly. Aldous has termed this {\it extended weak convergence} \cite{aldous1981weak} and Hellwig has named it {\it the information topology} \cite{hellwig1996sequential}; these have recently been shown to be equivalent in discrete-time \cite{backhoff2019all,pammer2024note}. The closely related {\it adapted Wasserstein metric} \cite{bartl2024wasserstein,backhoff2019all,beiglbock2022approximation} has been shown to possess strong robustness properties in a variety of applications \cite{bayraktar2020continuity,julio2020adapted,bartl2023sensitivity}, not unlike the closely related continuous weak convergence notion we have studied above for MDPs \cite{KaYu20} and in controlled diffusion problems \cite{pradhan2022near}. We should note one distinction, however: The convergence of stochastic processes is with regard to the process measure itself and not only on the conditional probability measures at each element in the domain; in particular, a prior measure may entail a more restricted support for future marginal measures and this may affect the convergence properties under the adapted Wasserstein distance whereas the convergence of a stochastic kernel sequence is with regard to the properties of the kernels defined on the entire (domain) space.
\end{remark}
}
\sy{
\subsection{Continuity and Approximations of Control Policies under Weak Kernel Topologies}

Under the Young narrow topology or $w^*$-topology on control policies, and hence under the weak kernel mean embedding topology (by Theorem~\ref{equivalence} and Theorem~\ref{tightness}), one obtains broad continuity and approximation results, as established in several foundational studies. Specifically, \cite{borkar1989topology,arapostathis2010uniform} showed that, in controlled diffusion processes, the expected cost depends continuously on the control policy, \cite{pradhan2022near,pradhanyuksel2023DTApprx,pradhanyuksel2024near} demonstrated that, since piecewise-constant or Lipschitz policies are dense under this topology (see, e.g., \cite{milgrom1985distributional,arapostathis2012ergodic}), both near-optimality of Lipschitz policies and the validity of discrete-time approximations follow. Related analyses on discrete-time approximations can also be found in \cite{fleming1976generalized,kushner2001numerical,kushner2014partial}.
Moreover, invariant measures of diffusions depend continuously on control policies under this topology \cite{arapostathis2010uniform}, and an analogous result holds in the discrete-time setting \cite{yuksel2023borkar}. By the equivalence between the Young topology and the weak kernel mean embedding topology, one can also establish existence and approximation results for systems with decentralized information structures \cite{YukselWitsenStandardArXiv,SaYu22,Sal20}. These findings further extend to mean-field and game-theoretic models, providing desirable compactness, convexity, and continuity properties \cite{mertens2015repeated,balder1988generalized,bayraktar2022finite}.

In summary, strong kernel mean embedding topology (with the $\infty$-norm) is well-suited for analyzing model convergence and robustness properties, whereas weak kernel mean embedding topology is particularly appropriate for control policy spaces, yielding broad results on continuity, existence, and approximation.
}

\section{Conclusion}

In conclusion, this paper presents a novel topology for stochastic kernels and explores its relationships with the established Young narrow topology and \( w^* \)-topology. We demonstrate that our new topology is equivalent to both the Young narrow topology and the \( w^* \)-topology on the set of stochastic kernels. While the \( w^* \)-topology and kernel mean embedding topology are relatively compact -- an essential feature for proving existence results -- they are not closed. In contrast, the Young narrow topology is closed but lacks relative compactness. Importantly, we identify that a tightness condition is necessary to achieve closure in the \( w^* \)-topology and the kernel mean embedding topology, as well as to ensure relative compactness for the Young narrow topology. Consequently, these three topologies provide equally valuable tools for establishing existence results. Moreover, the kernel mean embedding topology's Hilbert space structure facilitates the approximation of stochastic kernels using simulation data, a topic we elaborate on in the final section of the paper.

\section*{Acknowledgements}

The first author would like to thank Professor Aurelian Gheondea for introducing him reproducing kernel Hilbert spaces.


\end{document}